\documentclass[%
 reprint,
superscriptaddress,
 amsmath,amssymb,
 aps,
]{revtex4-1}

\usepackage{graphicx}
\usepackage{dcolumn}
\usepackage{bm}
\usepackage{physics}
\usepackage{multirow}
\usepackage{amsthm}
\usepackage{apptools}
\usepackage{chngcntr}
\usepackage{mathtools}
\usepackage{subcaption}
\usepackage{color}
\AtAppendix{\counterwithin{theorem}{section}}
\newtheorem{theorem}{Theorem}

\newtheorem{Lemma}[theorem]{Lemma}

\newtheorem{Proposition}[theorem]{Proposition}
\makeatletter
\def\blfootnote{\xdef\@thefnmark{}\@footnotetext}
\makeatother

\begin{document}

\preprint{APS/123-QED}

\title{Handling Leakage with Subsystem Codes}

\author{Natalie C. Brown}
\affiliation{School of Physics, Georgia Institute of Technology, Atlanta, GA 30332, USA}
\author{Michael Newman}
\email{michael.newman@duke.edu}
\affiliation{Departments of Electrical and Computer Engineering, Chemistry, and Physics, Duke University, Durham, NC, 27708, USA}
\author{Kenneth R. Brown}
\email{kenneth.r.brown@duke.edu}
\affiliation{School of Physics, Georgia Institute of Technology, Atlanta, GA 30332, USA}
\affiliation{Departments of Electrical and Computer Engineering, Chemistry, and Physics, Duke University, Durham, NC, 27708, USA}

\date{\today}
\begin{abstract}
Leakage is a particularly damaging error that occurs when a qubit state falls out of its two-level computational subspace.  Compared to independent depolarizing noise, leaked qubits may produce many more configurations of harmful correlated errors during error-correction.  In this work, we investigate different local codes in the low-error regime of a leakage gate error model.  When restricting to bare-ancilla extraction, we observe that subsystem codes are good candidates for handling leakage, as their locality can limit damaging correlated errors.  As a case study, we compare subspace surface codes to the subsystem surface codes introduced by Bravyi et al.  In contrast to depolarizing noise, subsystem surface codes outperform same-distance subspace surface codes below error rates as high as $\lessapprox 7.5 \times 10^{-4}$ while offering better per-qubit distance protection.  Furthermore, we show that at low to intermediate distances, Bacon-Shor codes offer better per-qubit error protection against leakage in an ion-trap motivated error model below error rates as high as $\lessapprox 1.2 \times 10^{-3}$. For restricted leakage models, this advantage can be extended to higher distances by relaxing to unverified two-qubit cat state extraction in the surface code. These results highlight an intrinsic benefit of subsystem code locality to error-corrective performance.

\end{abstract}

\maketitle


\section{Introduction}
Quantum information is subject to many different noise processes that depend upon the platform in which it is physically encoded.  Quantum error-correcting codes can be used to preserve such information for any length computation, given certain assumptions on the strength and correlations of the error processes \cite{Aliferis:2006, Knill:1996b}.  In this work, we address an error model pertinent to many qubit architectures, leakage, and identify subsystem codes well-suited to correct it.

\subsection{Leakage}
Leakage is a particularly damaging error that occurs when a qubit excites to a state outside of its two-level computational subspace \cite{byrd2004overview,wood2018quantification, wallman2016robust, wu2002efficient, byrd2005universal, wang2018adiabatic, jing2015nonperturbative, sun2017simple}.  Much work has been devoted to characterizing and reducing leakage at the experimental level in different qubit architectures, including photonics \cite{fortescue2014fault}, quantum dots \cite{fong2011universal, mehl2015fault, sala2018leakage, cerfontaine2016feedback, cerfontaine2019feedback, Chan:2019}, superconductors \cite{fazio1999fidelity, liang2005loss, tao2006numerical, ghosh2013understanding, mcconkey2017mitigating, ghosh2015leakage, mckay2017efficient, shim2016semiconductor, jerger2016realization, herrera2013tradeoff, terhal2019leakage}, topological qubits \cite{ainsworth2011topological}, and ion traps \cite{haffner2008quantum, brown2018comparing}.  Unlike erasure, the locations of leakage errors are unknown, and so it is far more damaging \cite{stace2009thresholds, fujii2012error, barrett2010fault}.  

Leakage rates are significant for several types of qubits \cite{brown2018comparing, strikis2018}, and can even be the dominant error rate in the system \cite{fortescue2014fault, andrews2018}.  However, even when leakage is relatively less likely than other types of noise, the havoc it wreaks at the level of error-correction can quickly make it a limiting error process \cite{brown2018comparing, suchara2015leakage}.  This is for two reasons.  First, standard error-correction is not equipped to handle errors acting outside the computational subspace, and so leakage can persist until it is corrected by some other means.  Second, before it is corrected, leakage can cause damaging patterns of correlated errors that depend on the specifics of a platform's gate implementations.  

Leakage can also present a fundamental error when compared to other noise sources.  For example, while improved controls are continually reducing the dominating effects of coherent errors in ion-trap technologies, scattering events that cause leakage are more difficult to subdue. Consequently, although lower leakage rates make it a lesser concern for certain architectures currently, handling leakage in the low-error regime is vitally important for performing long computations reliably in the future.

Fortunately, one can reduce leakage by regularly swapping or teleporting qubits to the computational subspace \cite{mochon2004anyon, suchara2015leakage}, and it has been rigorously shown that a threshold exists in this model \cite{aliferis2005fault}.  This introduces a balancing act.  Eliminating leakage too frequently incurs significant qubit and gate overhead, which in turn introduces more potential circuit faults.  Eliminating leakage too infrequently allows leaked qubits to interact with many other qubits, introducing higher-weight correlated errors.

\subsection{Local Codes}
In this work, we focus on leakage in geometrically local codes, which are leading candidates for implementing quantum memories. Included in this category are topological codes, which exhibit some of the highest thresholds, support minimal overhead syndrome extraction, and are naturally amenable to qubit architectures that prefer local interactions \cite{Dennis:2002}.

Local codes are also a natural candidate for coping with leakage, as the number of qubits that may interact with any one leakage event is limited.  Consequently, previous works have focused on quantifying the performance of topological \emph{subspace} codes in the presence of leakage, including the surface code \cite{brown2018comparing, suchara2015leakage, fowler2013coping, mehl2015fault, ghosh2015leakage} and $[[7,1,3]]$ Steane or color code \cite{fortescue2014fault}.  

The (subspace) surface code is a particularly enticing choice as, in many respects, it is the top-performing quantum memory and a quantum analogue of the repetition code \cite{fowler2012towards,bravyi2010tradeoffs}.  This analogy can be made precise, in the sense that the surface code nearly saturates the CSS zero-rate entropic bound $2H(p_{\text{thr}}) \leq 1$ and can be realized as the hypergraph product of the repetition code with itself \cite{calderbank1996good, li20182, tillich2014quantum, Yoder:2019}.  Among its desirable features is the use of a single ancilla to extract syndromes, relevant to both its high performance and square lattice locality \cite{Dennis:2002,Tomita:2014}.

However, the surface code still suffers from significantly reduced thresholds \cite{suchara2015leakage} and an effective distance reduction \cite{fowler2013coping,brown2018comparing} in the presence of leakage.  It is then natural to ask if other local code families may perform better in certain regimes of a leakage error model.  

\subsection{Subsystem Codes}

Diverging from previous works, our focus will be on \emph{subsystem} codes, and in particular, subsystem surface and Bacon-Shor codes.  Subsystem codes are a generalization of subspace codes in which certain logical degrees of freedom are treated as gauges, and allowed to vary freely \cite{Bacon:2006, kribs2006operator}. Unlike topological subspace stabilizer codes, which require at least $4$-local interactions in a nearly Euclidean lattice \cite{aharonov2011complexity}, topological subsystem codes can be realized with $3$- or even $2$-local interactions \cite{bravyi2013subsystem, Bombin:2010b, bombin2012universal, suchara2011constructions}.  If one sacrifices topological protection, then non-topological local codes like the Bacon-Shor family can still yield good error protection at reasonable error rates with only $2$-local interactions and minimal overhead \cite{Napp:2012, brooks2013fault}.

Generally, subsystem codes come with certain intrinsic advantages.  Their increased locality is even more desirable for physically local qubit architectures.  They allow for simpler and parallelized syndrome extraction, which mitigates the circuit-depth required to measure higher-weight stabilizers \cite{Aliferis:2007}.  Geometrically constrained codes can yield better encoding rates compared to their subspace cousins \cite{Bravyi:2010, Yoder:2019, bravyi2010tradeoffs}. Finally, subsystem codes can yield tremendous advantages for implementing simple and universal fault-tolerant logic \cite{Paetznick:2013, Bombin:2013, yoder2017universal}, bypassing the high cost of magic-state distillation \cite{bravyi2005universal}.

However, the advantages of subsystem code locality often come at a high cost to their error-corrective performance.  While gauge degrees of freedom are useful for implementation in local architectures and fault-tolerant computation, they introduce more locations for potential circuit faults without giving additional information about the errors that are occurring.  This generally manifests as less reliable high-weight stabilizers and lower thresholds.

\subsection{Contributions}

In this work, we highlight a simple and intrinsic error-corrective advantage for these gauge degrees of freedom: limiting correlated errors due to leakage.  This differs from an independent depolarizing model, in which local measurements do not reduce the set of correlated errors.  

To be explicit, suppose we wanted to measure a stabilizer $S=G_1 \ldots G_\ell$, where each $G_i$ is a gauge operator.  If we measure each gauge separately, then a single fault while measuring $G_i$ will produce an error supported entirely on the support of $G_i$.  If instead we measure the entire stabilizer directly, then a single fault could produce an error supported on the support of many $G_i$. However, in an independent depolarizing model, as long as we measure along the gauge operators, this error will take the form $EG_{i+1} \ldots G_n$ where $E$ is supported entirely on the support of $G_i$.  In particular, the set of correlated errors is the same with or without local measurements up to a symmetry of the logical subsystem.  For some gate-error models with non-local interactions, it is even advantageous to measure the larger stabilizers directly \cite{Li:2018, li20182}. 

However, this is not true for typical leakage models, as a single leakage can produce a much richer set of correlated errors. Consequently, circuits that are fault-tolerant in a depolarizing model may no longer be fault-tolerant in a leakage model.  We illustrate this difference in the single ancilla extraction of a Bacon-Shor code stabilizer with both Pauli and leakage faults in Figure~\ref{Bacon-Shor}.

\begin{figure}[htb!]
\includegraphics[width=\linewidth]{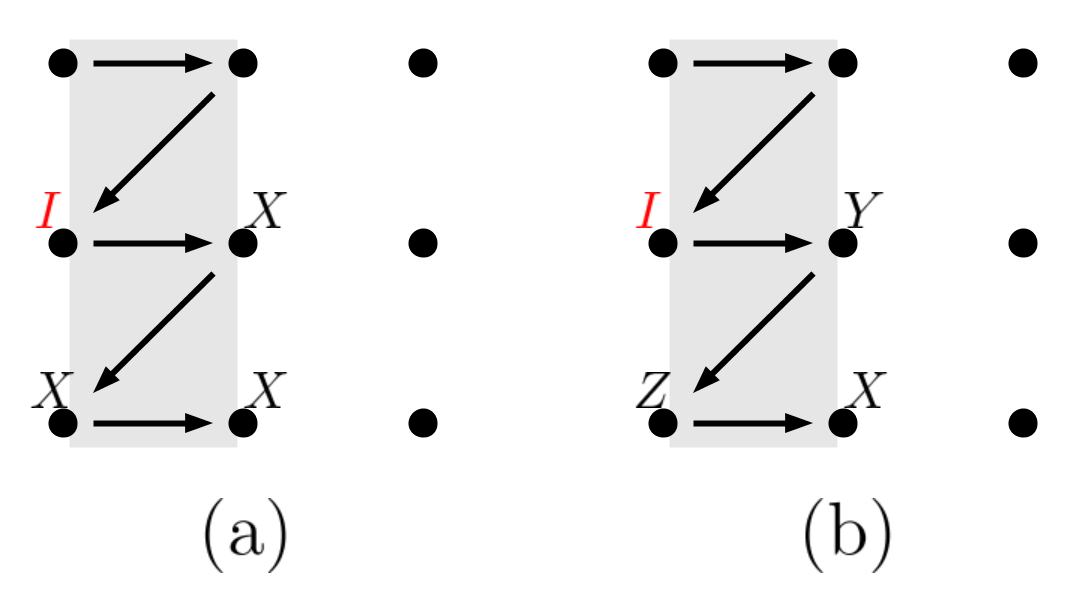}
\caption{Syndrome extraction in a vertical distance-$3$ $X$-type Bacon-Shor stabilizer (shaded) using a single ancilla \cite{Li:2018}, with horizontal $XX$ and vertical $ZZ$ gauge operators.  Implicit in the diagram is a Pauli $X$-error on the ancilla in (a), and a leaked ancilla in (b), both occurring after the third gate during extraction. The correlated error in (a) is equivalent to a single data qubit $X$-error up to gauge transformation.  If, for example, the leaked ancilla in (b) depolarizes each qubit it interacts with, then it may produce a damaging higher-weight correlated error like the one shown.}
\label{Bacon-Shor}
\end{figure}

We compare three code families in the presence of leakage: subspace surface codes, subsystem surface codes, and Bacon-Shor codes, while restricting to bare-ancilla syndrome extraction.  We consider the first two in both standard and rotated lattice geometries, and under various leakage reduction techniques.  In stark contrast to depolarizing noise, we find that subsystem surface codes actually outperform their same-distance subspace cousins at equal leakage-to-depolarizing ratios in the low-error regime.  Furthermore, we observe that low to intermediate distance Bacon-Shor codes outperform subspace surface codes in per-qubit leakage protection at similar error rates.  These advantages are due essentially to the subsystem codes' handling of uncontrolled `hook' errors during syndrome extraction with leaked qubits. Finally, if one allows for larger ancilla states in a restricted leakage model, we observe that intermingling minimal leakage reduction with partial fault-tolerance in the subspace code ultimately yields the best overall low-error performance. Our central contributions are the following.

\begin{enumerate}
    \item[$(i)$]
    We motivate a Pauli-approximated ion-trap leakage model, and demonstrate significantly improved performance compared to a worst-case stochastic leakage model. This model is similar to a photonics leakage model \cite{fortescue2014fault}.
    \item[$(ii)$]
    We show that any topological subspace stabilizer code must incur a linear reduction in effective distance when using minimal leakage reduction to correct worst-case stochastic leakage.  In the case of subspace surface codes, the effective code distance is halved \cite{fowler2013coping, brown2018comparing, suchara2015leakage}.
    \item[$(iii)$]
    We define a rotated variant of the subsystem surface codes proposed in \cite{bravyi2013subsystem} which yields reduced qubit overhead. We argue that these codes (along with their un-rotated variant) yield better per-qubit distance protection than subspace surface codes in the presence of leakage.  We give coarse upper-bounds on the error regime in which this advantage manifests, with $p \lessapprox 2.5 \times 10^{-4}$ with worst-case stochastic leakage and $p \lessapprox 0.32 \times 10^{-4}$ with ion-trap leakage.
    \item[$(iv)$]
    We provide Monte Carlo simulations of both the subspace and subsystem surface codes at low error rates using no qubit overhead for leakage elimination.  We find that subsystem surface codes outperform same-distance, same-geometry subspace surface codes for error rates $\lessapprox 0.75 \times10^{-3}$ with worst-case stochastic leakage.  This is increased to $\lessapprox 2.0 \times 10^{-3}$ for ion-trap leakage.
    \item[$(v)$]
    We additionally compare Bacon-Shor codes with serialized syndrome extraction to subspace surface codes with ion-trap leakage.  We observe that at low distances, Bacon-Shor codes per-qubit outperform surface codes in the $10^{-4}-10^{-3}$ error range.  This advantage persists even when the depolarizing rate is many times stronger than the leakage rate. 
    \item[$(vi)$]
    Finally, we compare different leakage reduction techniques in an ion-trap leakage model of the surface code while allowing for larger ancilla states.  We observe that combining unverified two-qubit cat state extraction with minimal leakage reduction yields an orders-of-magnitude improvement over existing fault-tolerant proposals \cite{suchara2015leakage}, and also begins to relatively outperform non-fault-tolerant strategies in the $\lessapprox 10^{-3}$ error range.  This advantage may extend to higher distances, but relies unavoidably on an assumption of independent leakage events.
\end{enumerate}

At the heart of these results is a simple accounting of the correlated errors that may appear due to leakage of both data and ancilla qubits, and the quantification of the regime in which their effects dominate.  However, it underscores an advantage of local subsystem error-correction that may extend to other non-Markovian error models, with local baths that are determined by the qubit connectivity \cite{terhal2005fault}. In these cases, local checks minimize what can go wrong.

The paper proceeds as follows.  In Section~\ref{Leakage}, we define the two leakage models we consider, and introduce the leakage reduction techniques we will apply.  In Section~\ref{Robustness}, we compare the leakage robustness of subspace surface codes and subsystem surface codes by accounting for fault-paths generated by leakage.  Although we focus on low-overhead leakage reduction with bare-ancilla measurements, we touch on more general leakage reduction strategies in Appendix \ref{auxiliary}. In Section~\ref{Simulations} we provide numerics comparing subspace and subsystem codes under various leakage reduction schemes and constraints.  Finally, we conclude with some discussion and potential avenues for improvement in Section~\ref{Discussion}.  Additional technical details concerning the simulations and gate schedulings may also be found in Appendix \ref{scheduling}. 

\section{Leakage}\label{Leakage}

Qubits are not completely isolated two-level systems, and sometimes they have energetic access to other states.  Leakage occurs when a qubit state excites out of its computational subspace. 

We call a qubit \emph{leaked} if its state is supported on the subspace $\mathcal{H}_{\ket{2}}$, where $\ket{2}$ represents the leaked degree of freedom.  Otherwise we call the qubit \emph{sealed}.  Note that there may be many leakage states, and leakage need not take this restricted form, but for simplicity we group these into a single $\ket{2}$-state.

Although leakage may be modeled as a coherent process \cite{ghosh2013understanding}, we make a standard simplifying assumption \cite{suchara2015leakage} that leakage occurs stochastically via the channel $$\mathcal{E}_{\ell}(\rho) = (1-p_{\ell})\rho + p_{\ell}\sum\limits_{i=0}^2 \ket{2}\bra{i}\rho\ket{i}\bra{2}.$$
We refer to $p_{\ell}$ as the \emph{leakage rate} of the system.  This simplified error channel ensures that we only consider mixtures of leakage states and sealed states, which will be necessary for efficient simulation.  Additionally, we assume that each qubit undergoes a relaxation channel whenever it could also leak,
$$\mathcal{E}_{r}(\rho) = (1-p_{r})\rho + p_r\left(P_C\rho P_C + \frac{1}{2}\sum\limits_{i=0,1} \ket{i}\bra{2}\rho\ket{2}\bra{i}\right),$$
where $P_C = \ket{0}\bra{0} + \ket{1}\bra{1}$ is the projector onto the computational space, and $p_{r}$ is the \emph{relaxation rate}.  For simplicity, we assume $p_{r} = p_{l}$.

In addition to leakage, we also assume background depolarizing noise.  In order to probe the leakage behavior of the codes, we do not want the depolarizing noise to drown out leakage effects. However, not including depolarizing noise introduces peculiarities in the error model, such as perfect measurement.  Consequently, we choose our \emph{depolarizing rate} $p_{d}$ to be the same as our leakage rate $p_{\ell}$, unless otherwise specified.  An equal leakage-to-depolarizing ratio is directly applicable in certain systems \cite{brown2018comparing,andrews2018, fortescue2014fault}, but even with depolarizing noise that is many times stronger, leakage may remain a prohibitive error source \cite{brown2018comparing, suchara2015leakage}.  We define the \emph{error rate} of the system to be $p = p_{\ell} + p_{d}$, as the leakage and depolarizing error processes occur independently.

Our gate error model involves three fundamental operations: preparation into $\ket{0}$ or $\ket{+}$, a two-qubit CNOT gate, and measurement in the $Z$- and $X$-bases. It is important to note that our gate error model does not include idling errors, a choice we make to simplify gate scheduling.  While idling errors do not usually cause leakage, one must be careful to consider more careful parallelization in architectures with shorter coherence times.  

Preparation can fail by wrongly preparing an orthogonal state with probability $p_{d}$.  Additionally, a newly-prepared qubit can immediately leak with probability $p_{\ell}$; for some systems, this is a pessimistic assumption.

A two-qubit CNOT gate can fail by applying one of the $15$ non-identical two-qubit Pauli operators to its support, each with probability $p_{d}/15$.  For our simulations, we assume that leakage events occur with probability $p_{\ell}$ independently on the support of the gate.  This is well-motivated in some architectures \cite{suchara2015leakage}, but less so in others.  However, the subsystem code constructions can also handle correlated leakage events with small modification, which we discuss in Appendix~\ref{correlated}.   

Finally, measurements can return the wrong result with probability $p_{d}$. Whenever a leaked qubit is measured, it will always return eigenvalue $-1$.  Implicitly, this assumes two-level measurements that cannot distinguish between excited states, and is also relevant to many qubit architectures \cite{suchara2015leakage, jerger2016realization}.  However, when accounting for fault-paths, we consider a worst-case scenario in which leaked qubits produce randomized outcomes when measured.

\subsection{Leakage Models}

There is some freedom in extending the dynamics of the system to interactions between leaked qubits and sealed qubits. In particular, we assume that our two-qubit gates are themselves \emph{sealed}.  For $\mathcal{H}_C$ the computational Hilbert space and $\mathcal{H}_{\ket{2}}$ the leakage Hilbert space, this means that any two-qubit gate $U$ factors as
$$U=U_{\mathcal{H}_C\otimes \mathcal{H}_C} \oplus U_{\mathcal{H}_C\otimes \mathcal{H}_{\ket{2}}} \oplus U_{\mathcal{H}_{\ket{2}} \otimes \mathcal{H}_C} \oplus U_{\mathcal{H}_{\ket{2}} \otimes \mathcal{H}_{\ket{2}}}.$$  This restriction is physically motivated in many architectures, including both superconductors \cite{suchara2015leakage, fowler2013coping} and ion-traps \cite{brown2018comparing}.  In particular, $U$ does not directly propagate one leaked qubit into two leaked qubits.

The two-qubit gate $U_{\mathcal{H}_C \otimes \mathcal{H}_C}$ is the computational gate and $U_{\mathcal{H}_{\ket{2}} \otimes \mathcal{H}_{\ket{2}}}$ can only apply a harmless phase.  However, it remains to define $U_{\mathcal{H}_C \otimes \mathcal{H}_{\ket{2}}}$ and $U_{\mathcal{H}_{\ket{2}} \otimes \mathcal{H}_C}$, the interactions between leaked and sealed qubits.

\subsubsection{Depolarizing Leakage}

We consider two such definitions, resulting in two different leakage models.  The first is a worst-case stochastic model that has been considered in previous works \cite{suchara2015leakage,fowler2013coping}.  We call this model \emph{depolarizing} or $DP$- leakage, and it is defined symmetrically for both $U_{\mathcal{H}_C\otimes \mathcal{H}_{\ket{2}}}$ and $U_{\mathcal{H}_{\ket{2}} \otimes \mathcal{H}_C}$.  In the $DP$-leakage model, the leaked qubit remains leaked while a Haar random unitary is applied to the sealed qubit. This results in complete depolarization of the sealed qubit, and is the more damaging of the two leakage models.

\subsubsection{M{\o}lmer-S{\o}rensen Leakage}

The second is a more restrictive model in which a leaked qubit simply does not interact with a sealed qubit.  This presents an effective gate-erasure error, which we expect may apply to several different architectures.  One such example is the standard entangling operation for ion-traps: the M{\o}lmer-S{\o}rensen gate \cite{molmer1999multiparticle}.  

The two-qubit M{\o}lmer-S{\o}rensen interaction is mediated by the shared motional modes in an ion-trap.  The underlying Hamiltonian is a state-dependent displacement, which is based on near-resonant driving of motional sidebands that correspond to removing or adding phonons of motion with excitation of the internal ion.  The Hamiltonian parameters are chosen to remove any residual entanglement between the ion and motion at the end of the gate. The ion-ion entanglement arises from a second-order term due to the time-dependent displacement not commuting with itself in time.  

M{\o}lmer-S{\o}rensen gates are often optimized in a way that is independent of the magnitude of the driving field, with the residual entanglement minimized ion by ion, and do not require a second ion to be driven.  For a leaked ion, if the spacing between the motional modes is small compared to the separation between the computational energy state and the leakage energy state, then the lasers that drive the displacement are far from resonant and both to the same side of the carrier transition.  The result is that the leaked ion is only very weakly displaced and therefore does not generate an entangling gate with the sealed ion, similar to a photonics model \cite{fortescue2014fault}.  

The M{\o}lmer-S{\o}rensen gate does not occur and the sealed ion is unaffected. For the full CNOT gate, the net effect is that different single qubit Pauli rotations are applied to the control and target qubits, see Figure~\ref{Molmer-Sorensen}.  
\begin{figure}[htb!]
\includegraphics[width=\linewidth]{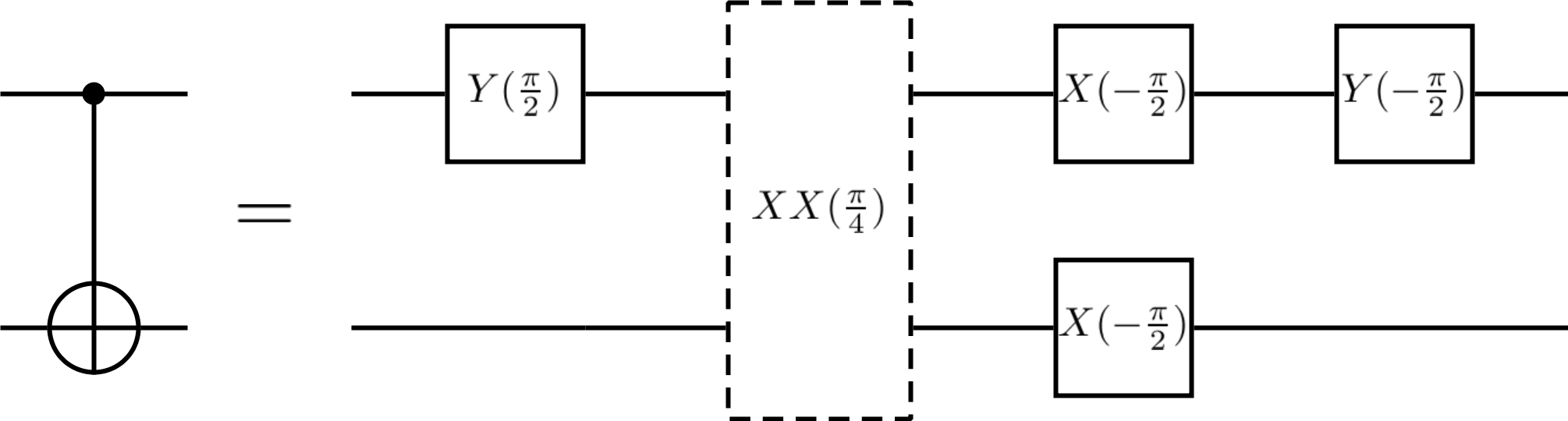}
\caption{An ion-trap CNOT gate, expressed as a product of native Pauli rotations. In the presence of a leaked qubit, the dashed M{\o}lmer-S{\o}rensen gate is not applied.}
\label{Molmer-Sorensen}
\end{figure}
The target undergoes an $X(-\pi/2)$-rotation along the Bloch sphere, while the control undergoes an $Z(-\pi/2)$-rotation.  Whichever qubit is leaked is unaffected by the sealed single-qubit gate.  In order to simulate leakage in the Pauli frame model, we replace these channels with their Pauli-twirl approximations.  This yields the stochastic error channels
\begin{align*}
\mathcal{E}_{bit}(\rho) &= \frac{1}{2}\rho + \frac{1}{2}X\rho X \text{, and} \\
\mathcal{E}_{phase}(\rho) &= \frac{1}{2}\rho + \frac{1}{2}Z\rho Z,
\end{align*}
where $\mathcal{E}_{bit}$ is applied to the target when the control is leaked, and $\mathcal{E}_{phase}$ is applied to the control when the target is leaked.  We refer to this model as \emph{M{\o}lmer-S{\o}rensen} or $MS$- leakage.  The twirl of this restricted model will yield significantly improved performance compared to $DP$-leakage, although one must be careful to weigh this against the generic benefit of twirling \cite{magesan2013modeling, iyer2018small}.

\subsection{Leakage Reduction}

As errors preserving the computational subspace accumulate, error-correction will periodically remove them.  However, standard error-correction does nothing to remove leakage.  Every leakage event will eventually relax back to the computational subspace, but these long-lived leakage errors will corrupt the surrounding qubits for several rounds of error-correction \cite{fowler2013coping}.  Without an active approach to remove leakage, it will completely compromise the efficacy of the code, see Figure~\ref{No-LRU}.

\begin{figure}[htb!]
\includegraphics[width=.8\linewidth]{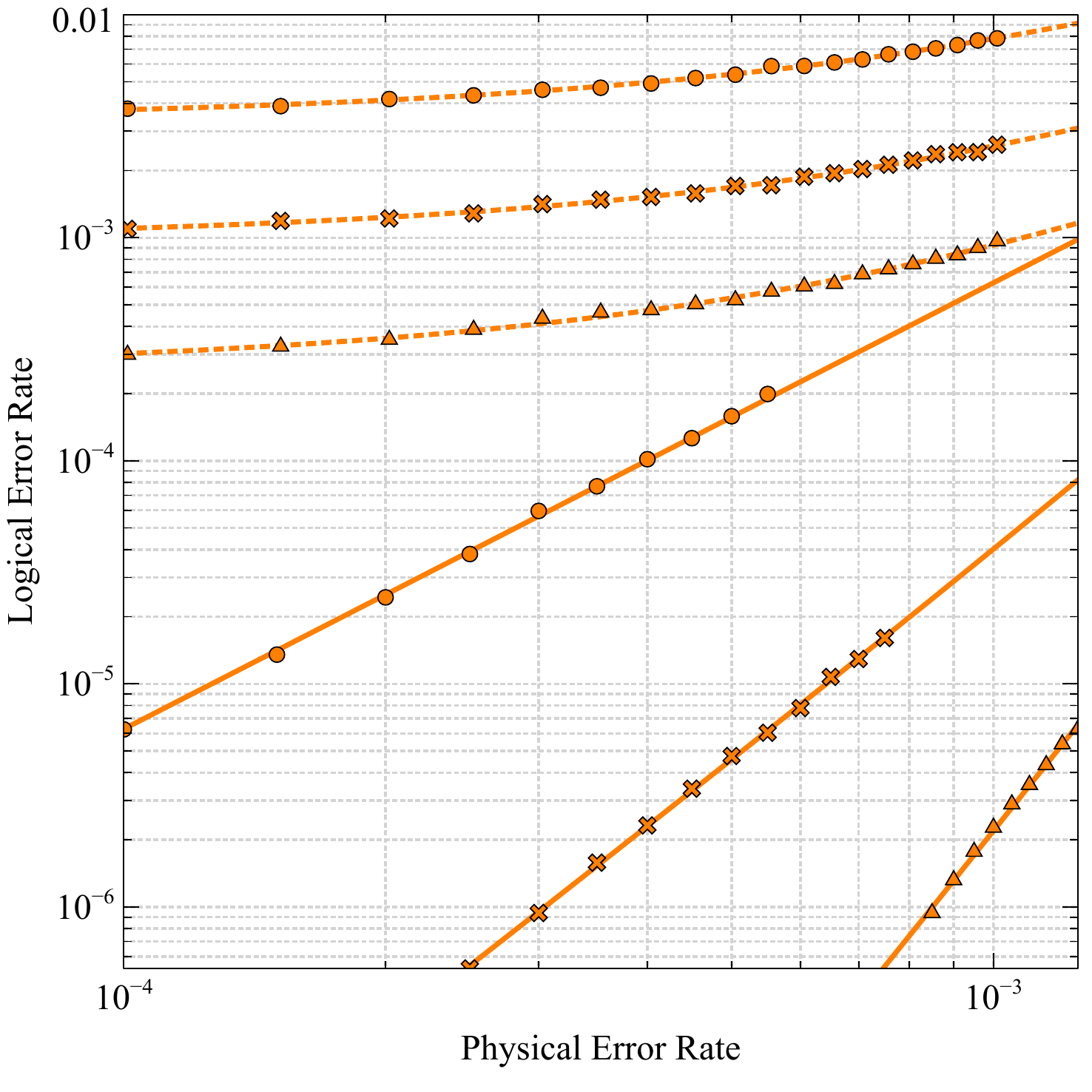}
\caption{Distance $3$,$5$, and $7$ standard surface code performance both with depolarizing leakage (dashed)  at $p_{d} = p_{r} = 100 p_{\ell}$ and without depolarizing leakage (solid) while using no leakage reduction. Spontaneous relaxation back to the computational subspace is an insufficient mechanism for leakage reduction even when leakage is orders of magnitude less likely than depolarization, and active methods are required.}
\label{No-LRU}
\end{figure}

\emph{Leakage reduction units}  or LRUs are circuit gadgets that are used to actively remove leakage \cite{mochon2004anyon, aliferis2005fault}.  They are defined by two properties:
\begin{enumerate}
    \item[$(i)$] if the quantum state is sealed, an LRU ideally acts as the identity on that state, and
    \item[$(ii)$] if the quantum state is leaked, then an LRU ideally projects the state back to the computational subspace.
\end{enumerate}

There are two popular approaches for eliminating leakage.  The first is to introduce auxiliary qubits and regularly swap between them and the initial data qubits in LRUs.  After the states are swapped, the initial data qubits are measured and reprepared in the computational space, removing leakage.  If one relaxes the assumption that the gates are sealed, then one must teleport the initial data qubits to remove leakage, see Figure~\ref{LRU}.  The frequency and placement of these LRUs will determine the code performance, which we address in Section \ref{Robustness}. 

The second approach foregoes auxiliary qubits, and instead periodically swaps the roles of data and ancilla in the code. This ensures that each physical qubit is measured in every other round \cite{mehl2015fault, brown2018comparing, suchara2015leakage}.  This technique requires no qubit overhead, and may preserve the locality of the qubit lattice.  For these reasons, it may be the more desirable approach.  While it is more complicated to analyze the effects of longer-lived leakage faults, we study this technique numerically in Section \ref{Simulations}. 

\begin{figure}[htb!]
\includegraphics[width=0.75\linewidth]{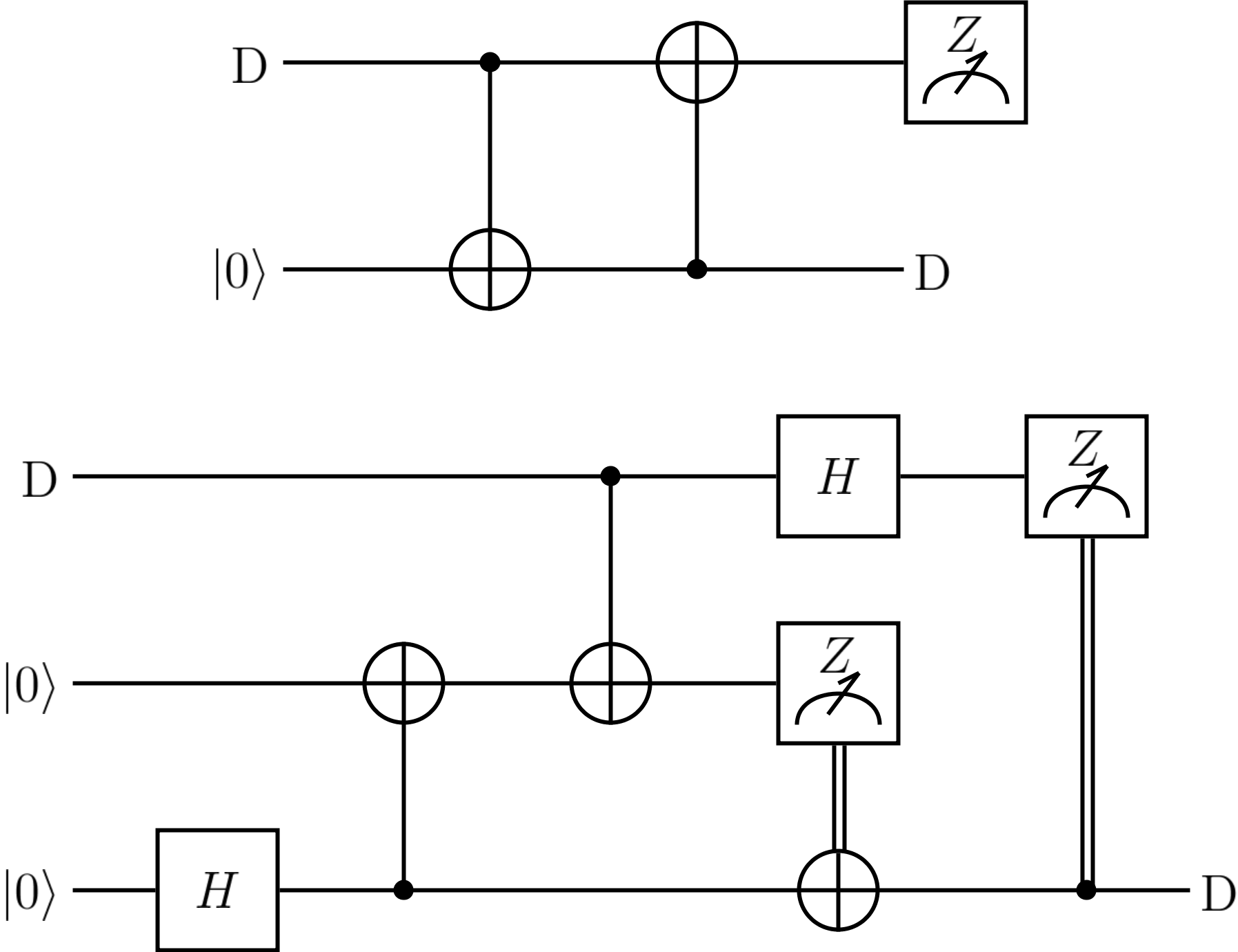}
\caption{Leakage can be removed via LRUs constructed from swapping with an auxiliary qubit (top) or teleportation (bottom).  In the ideal case, the data $D$ is teleported to the bottom wire.  Note that a single leakage cannot persist through the LRU in both the sealed-gate (top) and unsealed-gate (bottom) settings.}
\label{LRU}
\end{figure}

How frequently should we insert LRUs into our error-correction circuit?  We focus on two minimal overhead strategies.  

\begin{enumerate}
    \item[$(i)$] \emph{Syndrome extraction leakage reduction} (syndrome-$LR$) swaps each data qubit with an associated auxiliary qubit at the end of each syndrome extraction.  This allows a single leakage to persist for a single syndrome extraction cycle.  With some circuit delay, the ancillae may be used as the auxiliary qubits.
    \item[$(ii)$] \emph{Swap leakage reduction} (swap-$LR$) removes leakage by applying a swap gate between each ancilla qubit and the last data qubit it interacts with in each syndrome extraction cycle.  This allows a single leakage to persist for at most two consecutive syndrome extraction cycles; see Figure \ref{swapping}.
\end{enumerate}

\section{Leakage Robustness}\label{Robustness}
In this section, we detail damaging correlated errors due to leakage, and argue about the effective distance of different surface codes using syndrome-$LR$.  We compare subspace and subsystem surface codes, in both standard and rotated lattice geometries, in the presence of both $DP$- and $MS$-leakage using bare-ancilla extraction.  While we focus on leakage faults using syndrome-$LR$, we give numerical evidence that the additional time-correlated errors introduced by swap-$LR$ are not too damaging.

In each case, we consider planar codes with boundary encoding a single qubit \footnote{For standard geometries, we approximate the code with boundary by a code with periodic boundaries for threshold computations, but only preserve a single logical qubit.}.  Intuitively, the increased locality of the subsystem codes limits the correlated errors that can occur due to leakage.  This comes at the expense of a larger qubit lattice to achieve the same code distance. We also allow lattice geometries that are either rotated (a periodic diamond cut) or standard (i.e. un-rotated, a periodic square cut).  Rotated surface codes give a $\sqrt{2}$-fold increase in code distance \cite{Tomita:2014,bombin2007optimal} per qubit. However, a higher fraction of fault-patterns may be more damaging, and recent work has shown that standard surface codes may outperform rotated surface codes within certain sub-threshold error regimes \cite{beverland2018role}.  A similar intuition holds for leakage, as a standard lattice geometry can ensure that leakage faults do not propagate errors parallel to a logical operator, again at the cost of additional qubit overhead. 

We call a code \emph{leakage robust} if it does not experience a linear effective distance reduction in the presence of leakage.  Otherwise, we call them \emph{leakage susceptible}, in which case the effective distance is halved in surface codes. To achieve a desired $d_{\text{eff}}$ in the presence of $MS$-leakage, we find that it is asymptotically optimal to choose a rotated subsystem code, and next a standard subspace code, both of distance $d_{\text{eff}}$.  In the presence of $DP$-leakage, we find that it is asymptotically optimal to choose a standard subsystem code of distance $d_{\text{eff}}$, and next a rotated subspace code of distance $2d_{\text{eff}}$.  

In both cases, the subsystem surface codes yield better per-qubit distance protection than subspace surface codes, as summarized in Table~\ref{Overhead}.  We expect this subsystem advantage to generalize to multiple encoded qubits and persist against any surface code geometry, as rotated $[[n,k,d]]$ surface codes already saturate the $2$-$D$ topological code bound $kd^2 \leq cn$, where $c \geq 1$ for a planar square lattice architecture with closed defects \cite{bravyi2010tradeoffs,delfosse2016generalized}. 

Better per-qubit distance protection implies that subsystem surface codes outperform subspace surface codes in the $p\rightarrow 0$ limit.  Although thresholds may also be reduced by the addition of local correlated data qubit errors, these entropic effects are relatively mild \cite{chubb2018statistical}.  Consequently, at higher error rates, subspace surface codes achieve better performance due to their higher threshold even if the effective distance is reduced. After analyzing correlated fault patterns using syndrome-$LR$ in the present section, we investigate this crossover point in same-distance codes using swap-$LR$ in Section \ref{Simulations} numerically.

It is worth noting that the insertion of additional leakage reduction can restore the effective distance of any code.  For example, performing a leakage reduction step after each individual gate immediately converts leakage errors to depolarizing errors as they arise.  However, this introduces enormous circuit-volume overhead, which results in significantly reduced performance and many extra qubits \cite{suchara2015leakage}.  Additionally, unlike swap-$LR$, it will necessarily increase the required connectivity of the qubit lattice.  Thus, we focus on minimal overhead leakage reduction, but discuss additional strategies in Appendix \ref{auxiliary}.

\begin{table}[htb!]
\centering 
\begin{tabular}{ c | c | c |} 
&$DP$-Leakage & $MS$-Leakage\\ [0.5ex] 
\hline 
 Rotated Subspace & $4d^2 + O(d)$ & $4d^2 + O(d)$ \\ 

Rotated Subsystem & $6d^2 + O(d)$ & {\color{red} $1.5d^2 + O(d)$} \\ 
 Standard Subspace & $8d^2 + O(d)$ & $2d^2 + O(d)$ \\ 
 Standard Subsystem & {\color{red} $3d^2 + O(d)$} & $3d^2 + O(d)$ \\ 
\hline 
\end{tabular}
\caption{Total number of data qubits required to realize an effective distance $d$ in the presence of different leakage models.  Optimal choices are highlighted in red. Accounting for ancilla qubits depends on tradeoffs between qubit reuse and parallelization. In the fully parallelized case, the subsystem advantage persists but is reduced; see Appendix \ref{auxiliary}.  \iffalse The effective distance $d_{\text{eff}}$ that can be realized given the $2d^2 - 1$ qubits required for a distance-$d$ rotated surface code as data+ancilla qubits for the new code. \fi} 
\label{Overhead} 
\end{table}

\subsection{Subspace Surface Codes}

We first consider fault-tolerance using syndrome-$LR$ on subspace surface codes with bare-ancilla syndrome extraction. Subspace surface codes are defined on a square lattice, with qubits placed on the edges of the lattice.  To each plaquette $P$ in the bulk of the lattice, we associate a $4$-body $X$-type stabilizer, and to each vertex $V$ in the lattice, we associate a $4$-body $Z$-type stabilizer,
$$X_P := \prod\limits_{e \in P} X_e \text{ and } Z_V := \prod\limits_{e \ni V} Z_e.$$ Then the stabilizer group is generated by all plaquette and vertex operators of this form, as well as those on the boundary.  For a standard lattice these boundary stabilizers are weight-$3$, while for rotated lattices they are weight-$2$, and we refer to violated stabilizers in the lattice as excitations.  Standard codes form a family with parameters $[[2d^2 - 2d + 1,1,d]]$, while rotated codes form a family with parameters $[[d^2,1,d]]$.

In the presence of $DP$-leakage, both the standard and rotated surface codes experience a halving of their effective distance, caused in part by uncontrolled `hook' errors \cite{Dennis:2002} during syndrome extraction \cite{fowler2013coping, brown2018comparing, suchara2015leakage}.  Three such distance-damaging faults, caused by ancilla leakage, are shown in Figure~\ref{subspace}.

\begin{figure}[htb!]
\includegraphics[width=\linewidth]{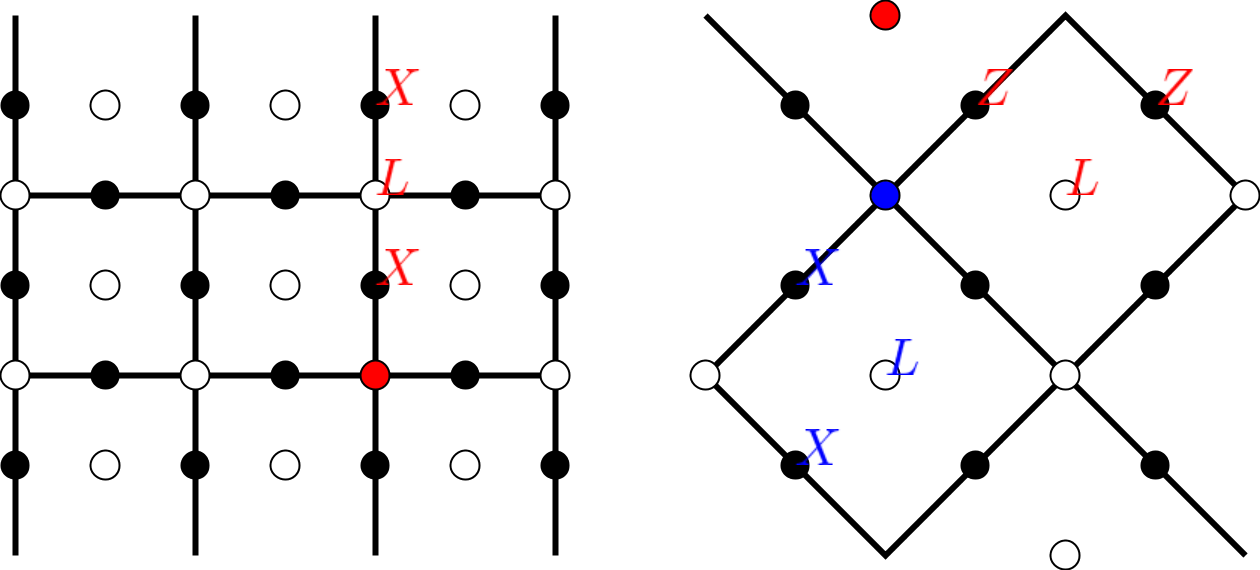}
\caption{Distance-damaging leakage faults in the distance-3 standard and rotated surface codes.  Red indicates $DP$-leakage; blue indicates $MS$-leakage, with anti-commuting excitations marked with the same color and leaked qubits measured as $+1$. The rotated surface code has north-south $X$-type boundaries, and east-west $Z$-type boundaries.}
\label{subspace}
\end{figure}

However, standard subspace codes are robust to $MS$-leakage. Unlike depolarizing errors, a single data leakage event may cause many different configurations of measurement outcomes on incident stabilizer checks before its removal.  Thus, a space-correlated error due to a data leakage may produce any one of its possible time-correlated syndrome configurations, but only over a single time-step.  Given $d$ successive syndrome measurements, this reduces the problem to considering only space-correlations in fault patterns.

In the $MS$-leakage model, data leakage is not too damaging: although it may generate many combinations of time-correlated syndrome configurations, it does not produce any new space-correlated errors.  The reason is that data $MS$-leakage cannot propagate errors to ancillae that will then propagate errors to other data qubits.  Thus, it suffices to only consider space-correlated errors due to ancilla leakage.

Ancilla $MS$-leakage may produce arbitrary configurations of errors on the support of a stabilizer that are of the same type as that stabilizer.  For the subspace surface code, every such configuration is either a single-qubit or two-qubit error, up to stabilizer equivalence. In particular, in the standard lattice, there are two different configurations of excitations (up to symmetry) caused by weight-$2$ correlated errors. Both are contained in an $L^\infty$-ball of radius one on the ancilla qubit lattice, and so any configuration of excitations caused by $<d$ faults cannot traverse the lattice. As any logical operator which does not traverse the lattice must be trivial, we may conclude that the effective distance of the code is preserved.  See Figure~\ref{good_subspace} for a summary of these worst-case leakage events.

\begin{figure}[htb!]
\includegraphics[width=0.8\linewidth]{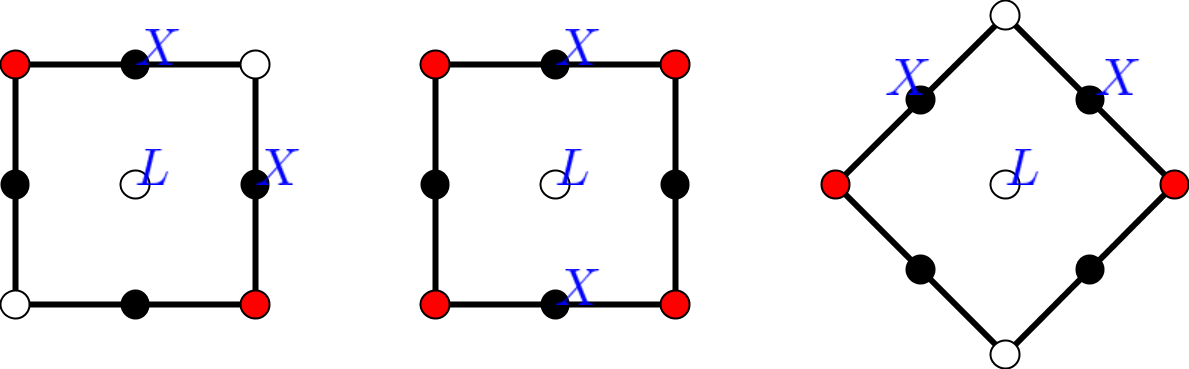}
\caption{Worst-case $MS$-leakage in subspace surface codes that does not cause an effective distance reduction, with $Z$-type excitations colored red.  Again, we assume north-south $X$-type boundaries of the rotated code.  Note that the diagonal errors in the standard lattice geometry may turn parallel to a logical operator in the rotated lattice geometry.}
\label{good_subspace}
\end{figure}

\subsection{Leakage Susceptibility of Subspace Codes}

As subspace surface codes are susceptible to $DP$-leakage, each experiences an effective distance reduction $d \mapsto \lceil\frac{d}{2}\rceil$ \cite{fowler2013coping, suchara2015leakage, brown2018comparing}.  This is damaging to their low-error suppression, and we will show that this damage begins to manifest at low but relevant sub-threshold error rates.  One might hope to construct some clever topological code and scheduling that minimizes the damage that $DP$-leakage inflicts on the effective code distance.

Unfortunately, a simple union bound shows that this damage is inevitable when restricting to minimal overhead leakage reduction in topological stabilizer \emph{subspace} codes.  However, the proof will give intuition for avoiding this damage by extending the search to topological stabilizer \emph{subsystem} codes. 

\begin{Proposition}[Union Bound]\label{main}
Let $\{C_d\}$ be a family of $D$-dimensional topological subspace stabilizer codes parametrized by diverging distance $d$ using bare-ancilla syndrome extraction and syndrome-$LR$.  Let $d_{\text{eff}}$ denote the minimum number of circuit faults that may cause a logical error in the presence of $DP$-leakage.  Then there must be a linear effective distance reduction, i.e. there is a constant $\eta<1$ such that $d_{\text{eff}} \leq \eta d.$
\end{Proposition}

\begin{proof}

By definition, each $C_d$ is defined on a $D$-dimensional lattice of qubits with a metric.  As the code family is topological, there must be some maximum diameter $\xi$ of the support of any stabilizer generator of the code, where $\xi$ is independent of $d$.

Select any weight-$d$ logical operator $$L = P_{\ell_1} P_{\ell_2} \ldots P_{\ell_d}$$ acting on qubits $\ell_1,\ldots,\ell_d$, where each $P \in \{X,Y,Z\}$.
Select a maximal set $R$ of disjoint neighborhoods $B_{\ell_i}(\xi)$ centered at qubits in the support of $L$.  Recall that $$\text{vol}(B_{\ell_i}(\xi)) = \frac{\pi^{D/2}}{\Gamma(\frac{D}{2}+1)}\xi^D =: c_D\xi^D.$$ Then there must be at least $\lfloor \frac{d}{c_D\xi^D} \rfloor$ such neighborhoods in this collection.

Using syndrome-$LR$, note that $r$ ancilla leakage events can cause any configuration of errors on the combined support of $r$ stabilizer generators.  Furthermore, for each $B_{\ell_i}(\xi) \in R$, there is some stabilizer generator $S_{\ell_i}$ such that $\{P_{\ell_i}, S_{\ell_i}\} = 0$.  As $[S_{\ell_i},L] = 0$, there must be at least one other $P_{\ell_j}$ supported in $B_{\ell_i}(\xi)$ with $\{P_{\ell_j},S_{\ell_i}\} = 0$.

In particular, each $B_{\ell_i}(\xi) \in R$ contains a stabilizer that intersects $L$ in at least two locations, $\ell_i$ and $\ell_j$.  Thus, a single ancilla leakage could produce the error $P_{\ell_i}P_{\ell_j}$, reducing the effective code distance by one. 

Now suppose we have two such stabilizer generators: $S_{\ell_i}$ corresponding to $P_{\ell_i}$, and $S_{\ell_i'}$ corresponding to $P_{\ell_i'}$.  Then $\ell_j$, as an element in the support of $S_{\ell_i}$ with diameter $\xi$, must lie in $B_{\ell_i}(\xi)$.  Similarly, $\ell_{j'}$ must lie in $B_{\ell_j'}$.  As these balls are disjoint, $\ell_j \neq \ell_{j'}$.

Thus, each ball in $R$ corresponds to an ancilla leakage which reduces the effective distance of the code by at least one, and these reductions combine so that $d_{\text{eff}} \leq d - |R|$.  As $|R| \geq \frac{d}{c_D \xi^D}$, we may conclude that $d_{\text{eff}} \leq \eta d$, where $\eta$ is given by $1 - \frac{1}{c_D \xi^D}$.
\end{proof}

Practically speaking, for most popular topological codes, the effective distance is halved.  Because $DP$-leakage is so damaging, one might expect that \emph{any} code family, when restricted to minimal overhead leakage elimination, would incur a linear distance reduction. However this is not the case: if one relaxes the practical restriction of a topological generating set, then we can manage $d_{\text{eff}} = d - 1$ in the surface code by overlapping measurement supports; see Appendix \ref{topcounterexample}.

A more plausible solution is to use subsystem codes, in which we can measure operators that anticommute with (dressed) logical operators. Such codes are natural to consider, as their increased locality may require less relative overhead to be robust in the presence of leakage.  In fact, we will see that subsystem surface codes provide a subsystem counterexample to Proposition \ref{main}.

\subsection{Subsystem Surface Codes}

\subsubsection{Constructions}

In this section, we construct subsystem surface codes from the perspective of ensuring robustness in the presence of leakage.  On a standard lattice, these codes are equivalent to those introduced in \cite{bravyi2013subsystem}, and have weight-$6$ stabilizers that can be expressed as the product of weight-$3$ gauge operators in the bulk.  We extend these codes to a rotated lattice, which have the same-weight bulk stabilizers and gauge operators.  Opposite the subspace surface codes, the standard lattice has weight-$2$ boundary operators while the rotated lattice has weight-$3$ boundary operators.  

Begin with the square lattice defining the surface code and insert a data qubit into the center of each plaquette.  This triangulates the square lattice on which the surface code was initially defined, doubling the distance of the code with respect to $Z$-type errors. Furthermore, measuring these newly formed triangular $X$-type stabilizers in the presence of leakage groups the original data qubits into `hooks', in which leakage can only recreate the hook errors defined in \cite{Dennis:2002}.  Unfortunately, this asymmetry between $X$ and $Z$ produces higher-weight, problematic hexagonal $Z$-type stabilizers.  Measuring these larger $Z$-type stabilizers directly will damage the code more than measuring the stabilizers defined on the original square lattice.  

The simple fix is to symmetrize the $X$- and $Z$-type operators: make both the $X$-plaquettes and the $Z$-plaquettes hexagonal.  As a result, we have $(d-1)^2$ gauge degrees of freedom that we may use to measure each hexagonal stabilizer as a product of triangular gauge operators.  Intuitively, as this groups the original data qubits into hooks for both $X$- and $Z$-type measurements, it should preserve the distance of the code in the presence of $DP$-leakage.  

The price for this locality (as with many subsystem codes) is more qubits and higher-weight stabilizers, which in turn yield higher logical error rates and lower thresholds.  This can be realized by relating the code capacity threshold to a phase transition of the random-bond Ising model on the honeycomb lattice along the Nishimori line \cite{nishimori1986geometry, bravyi2013subsystem}.  The resulting threshold estimate yields $p \approx 7\%$ \cite{de2006multicritical}, compared to the surface code threshold estimate on a square lattice of $p \approx 11\%$ \cite{Honecker:2000}. 

See Figure \ref{subsystem_code} for a pictoral description of these codes with boundary, which form a subsystem code family with parameters $[[3d^2-2d, 1, d]]$ and $(d-1)^2$ gauge degrees of freedom.
\begin{figure}[htb!]
\includegraphics[width=0.95\linewidth]{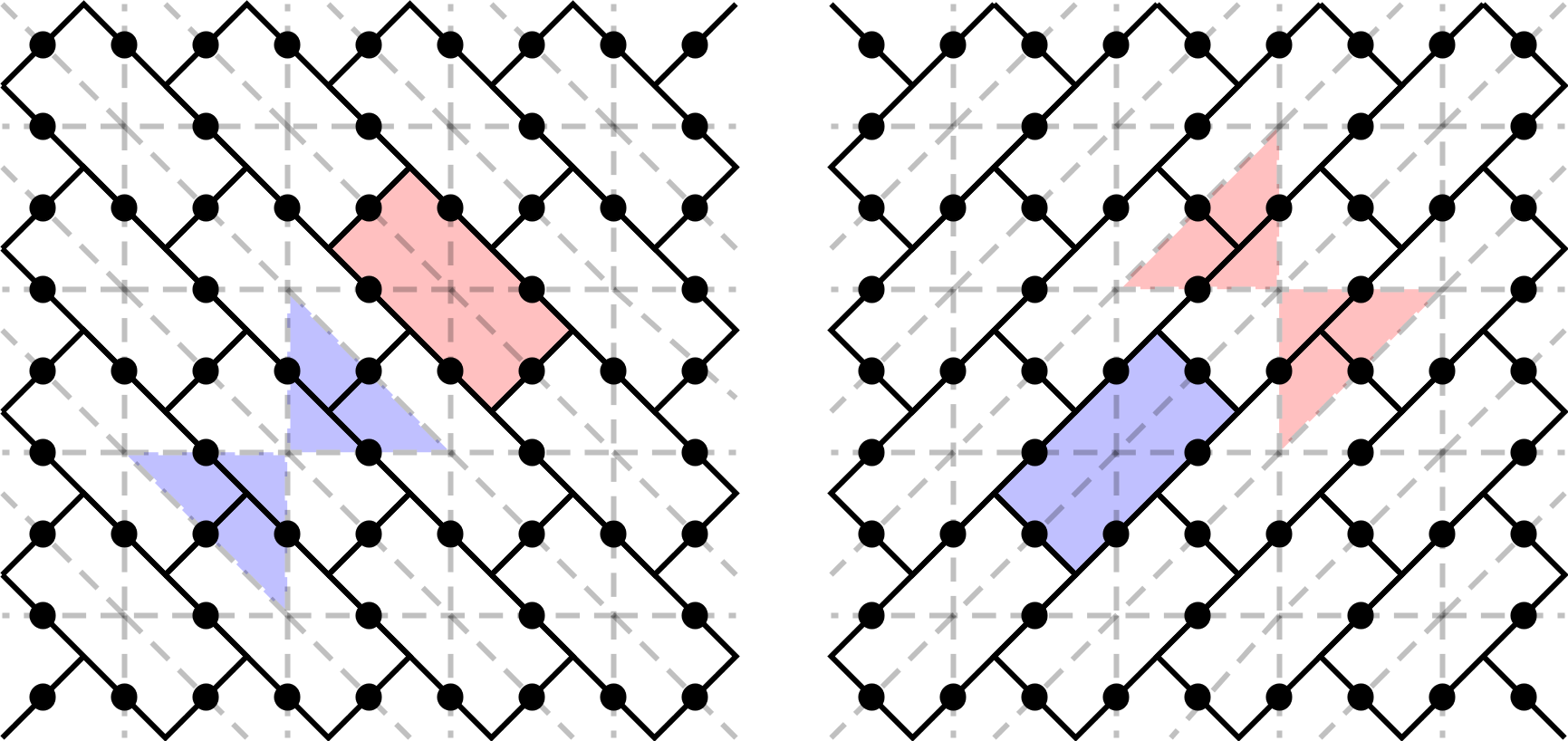}
\caption{A distance-$5$ subsystem surface code on a standard lattice.  The $X$-type stabilizers are defined by the hexagonal plaquettes on the left, and similarly for $Z$-type stabilizers on the right.  The dotted lines form the dual lattices, which represent gauge operators of opposite type.  In particular, the red $X$-type stabilizer may be realized as the product of the two red $X$-type gauges, and similarly for the $Z$-type operators in blue.  Boundaries can be assigned as weight-$2$ operators of the same type along opposite sides of the lattice.}
\label{subsystem_code}
\end{figure}
These codes inherit several nice properties from the subspace surface codes, including defect-based logical encoding, similar transversal gates, and efficient minimum-weight perfect matching decoding.  Unfortunately, these codes also have significant qubit overhead per distance.   For example, the smallest error-correcting code in the family forms a $[[21,1,3]]$ code.  

Fortunately, analogous to the surface code, we can rotate the lattice in order to reduce this overhead.  However, unlike the rotated surface code, the boundaries are fixed by the anisotropic orientation of the stabilizers.  This subsystem code family has parameters $[[\frac{3}{2}d^2 - d + \frac{1}{2}, 1, d]]$ with $\frac{(d-1)^2}{2}$ gauge degrees of freedom.  In particular, the smallest error-correcting code in this family forms an $[[11,1,3]]$ code with at most weight-$3$ check measurements, see Figure \ref{subsystem_small}.

\begin{figure}[htb!]
\includegraphics[width=0.9\linewidth]{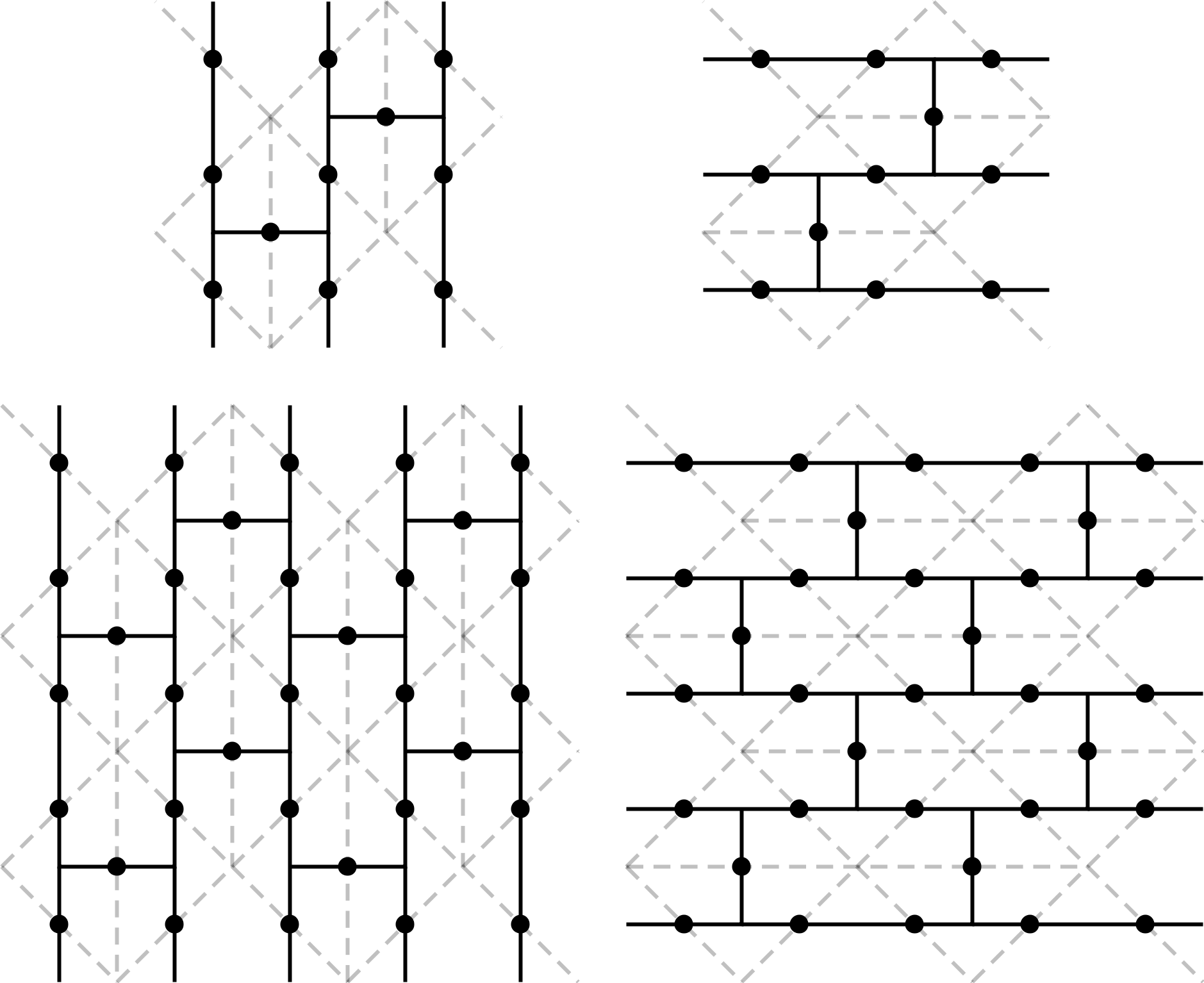}
\caption{Distance-$3$ (top) and distance-$5$ (bottom) subsystem surface codes on a rotated lattice.  The plaquettes on the left-side represent $X$-type stabilizers, while plaquettes on the right represent $Z$-type stabilizers.  Note that $X_L$ spans the lattice from north to south, while $Z_L$ spans the lattice from east to west.  Again, the dual lattices representing gauge operators of opposite type are outlined by the dotted lines.}
\label{subsystem_small}
\end{figure}

\subsubsection{Ancilla Leakage}
Having defined the subsystem surface code families, we turn to analyzing their correlated errors in the presence of leakage.  Time-correlated syndrome configurations may be treated as before, and so we may again restrict our attention to space-correlated errors.  We consider ancilla leakage and data leakage separately.

Ancilla leakage is much simpler to handle in this code, as the gauge generators are all weight-$3$.  Select any triangular $X$-gauge generator; by symmetry, the same analysis will apply to all other generators.  Then, up to gauge transformation, any error configuration of $X$-type will produce an effective weight one error on the data.  As data leakage in the $MS$-model also causes no new space-correlated error configurations, we may immediately conclude that these codes are robust to $MS$-leakage.  

Thus, we may focus solely on $DP$-leakage. In particular, only $Z$-type error configurations occurring on $X$-type gauge operators (and vice versa) may produce higher-weight correlated errors. Of the three possible weight-$2$ configurations of $Z$-errors on an $X$-gauge operator, two are equivalent to weight one $Z$-errors.  Thus, we need only consider two new errors in the presence of $DP$-leaked ancillae: the remaining weight-$2$ $Z$-error, and the weight-$3$ $Z$-error that acts on the entire triangular $X$-gauge operator; see Figure \ref{bad_configurations}.

\begin{figure}[htb!]
\includegraphics[width=\linewidth]{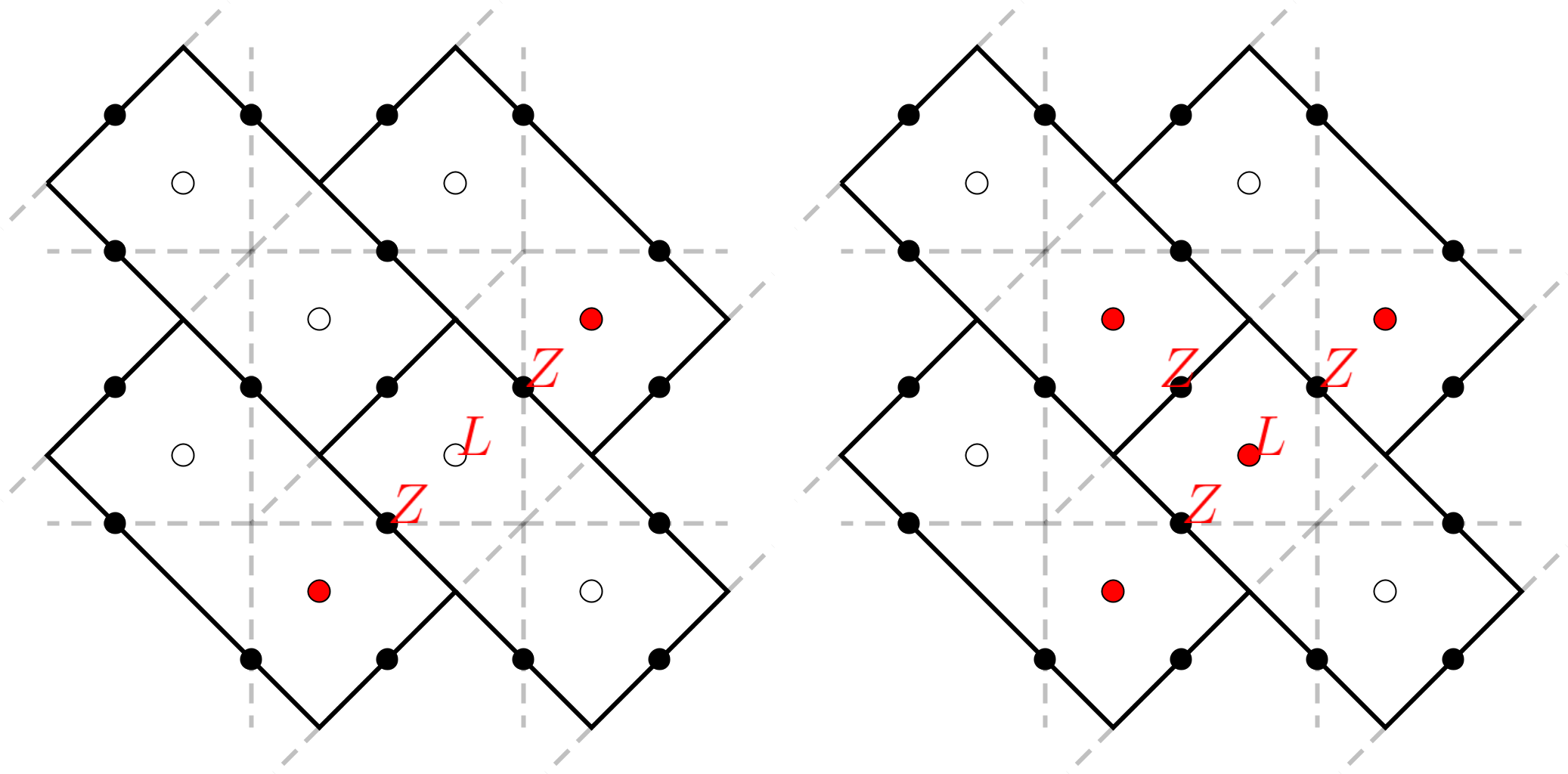}
\caption{The two new correlated $Z$-errors due to ancilla $DP$-leakage.  The dotted lines represent the gauge operators formed from the dual lattice of the $Z$-type stabilizers.  Red dots indicate violated gauge measurements.  A similar analysis applies symmetrically to every other triangular gauge operator.}
\label{bad_configurations}
\end{figure}

Note that the excitations formed from both errors in Figure \ref{bad_configurations} are contained in an $L^\infty$-ball of radius one in the standard lattice, but not in the rotated lattice.  Thus, neither of these errors may reduce the effective distance of the subsystem surface code in the standard lattice, while the first will halve the effective distance in the rotated lattice. 

Thus, it only remains to consider data leakage.  However, data leakage in the $DP$-model may cause error propagation between different data qubits.  Furthermore, this propagation will depend on the particular gate scheduling we choose.  Consequently, we relegate an accounting of data leakage and gate schedulings to Appendices \ref{myleakages} and \ref{scheduling}, respectively.  With proper gate timings, data leakage does not cause an effective distance reduction in these codes.

\subsection{Swap-$LR$}

One final type of leakage reduction proposed for the surface code in several works \cite{suchara2015leakage,ghosh2015leakage,mehl2015fault,fowler2013coping,brown2018comparing} is foregoing auxiliary qubits by regularly swapping the roles of data and ancilla.  The idea is to continually measure and reinitialize all qubits in the lattice to ensure that leakage can persist for no more than two rounds of syndrome extraction.  The mechanism for doing so is by applying SWAP gates between data and ancilla every round, see Figure \ref{swapping}.

\begin{figure}[htb!]
\includegraphics[width=0.85\linewidth]{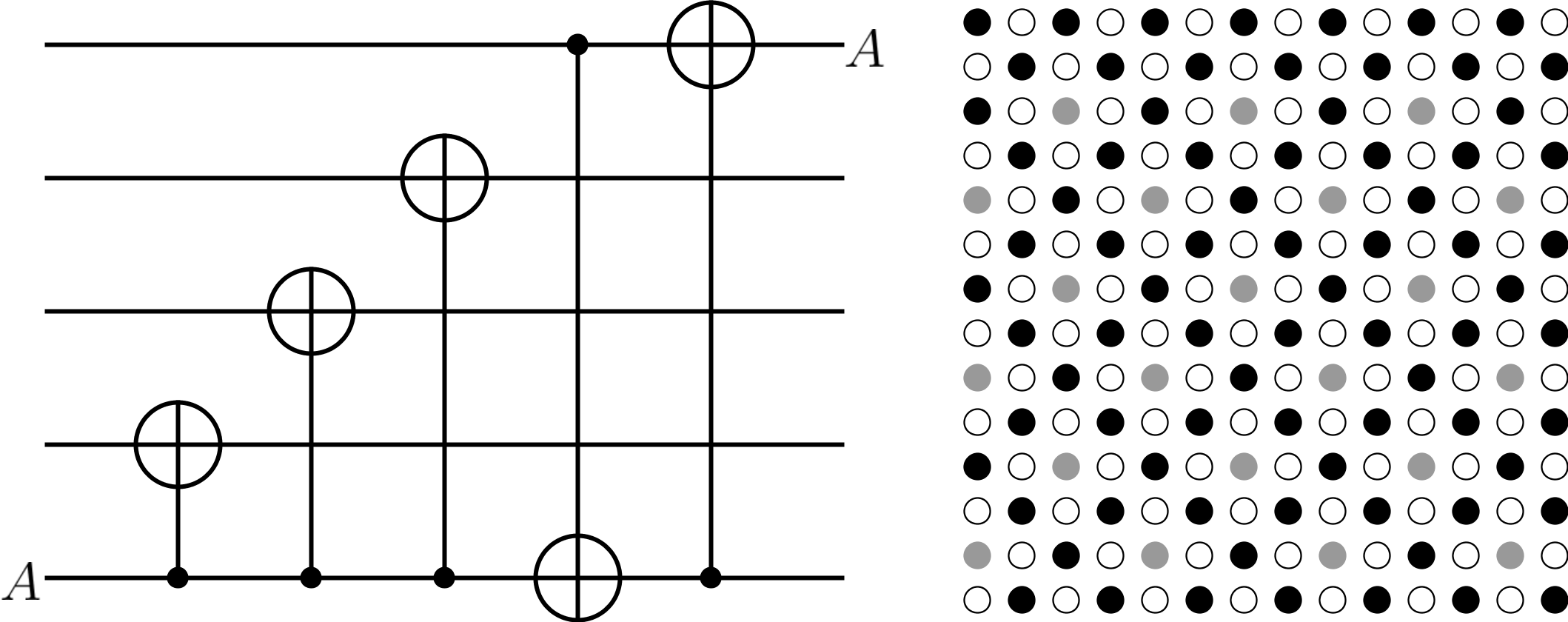}
\caption{The swap-$LR$ scheme switches the role of data and ancilla, with a subspace code scheduling pictured on the left.  One of the three CNOT gates forming the SWAP gate cancels with the final CNOT gate of syndrome extraction, leaving a single extra CNOT gate and no qubit overhead.  On the right, the bulk of the subspace surface code can be partitioned into code qubits (shaded black or grey) and ancilla qubits (unshaded) that switch roles.  Removing the shaded grey qubits yields the bulk of the subsystem surface code.}
\label{swapping}
\end{figure}

This has two competing effects.  On the one hand, it minimizes circuit overhead during syndrome extraction, yielding fewer potential fault locations.  On the other hand, it allows leakage to persist for longer, resulting in more correlated errors.  In the case of $DP$-leakage, the space-correlated errors it produces are also different: an ancilla leakage upon preparation depolarizes the support of its measurement, and then acts as a data leakage for one more round.  These new space-correlated errors can again be handled by careful gate scheduling, at the cost of introducing additional weight-$2$ errors that do not damage the effective code distance; see Appendices \ref{myleakages} and \ref{scheduling}.  

In the case of $MS$-leakage, data leakage can propagate no further space-correlated errors. Thus, in the presence of $MS$-leakage and without using a different gate scheduling, we expect swap-$LR$ to scale comparably to syndrome-$LR$.  The similarities in behavior were also reported in \cite{suchara2015leakage}.

This is precisely what we observe numerically by comparing the behaviors of $DP$-leakage and $MS$-leakage in a standard subspace surface code using swap-$LR$.  It appears that the longer-lived leakage errors are not much more damaging, as the full-distance scaling of the code is approximately preserved in the presence of $MS$-leakage within the error regime we consider, see Figure \ref{swap_numbers}.

\begin{figure}[htb!]
\centering
\begin{subfigure}[b]{0.5\textwidth}
   \includegraphics[width=0.8\linewidth]{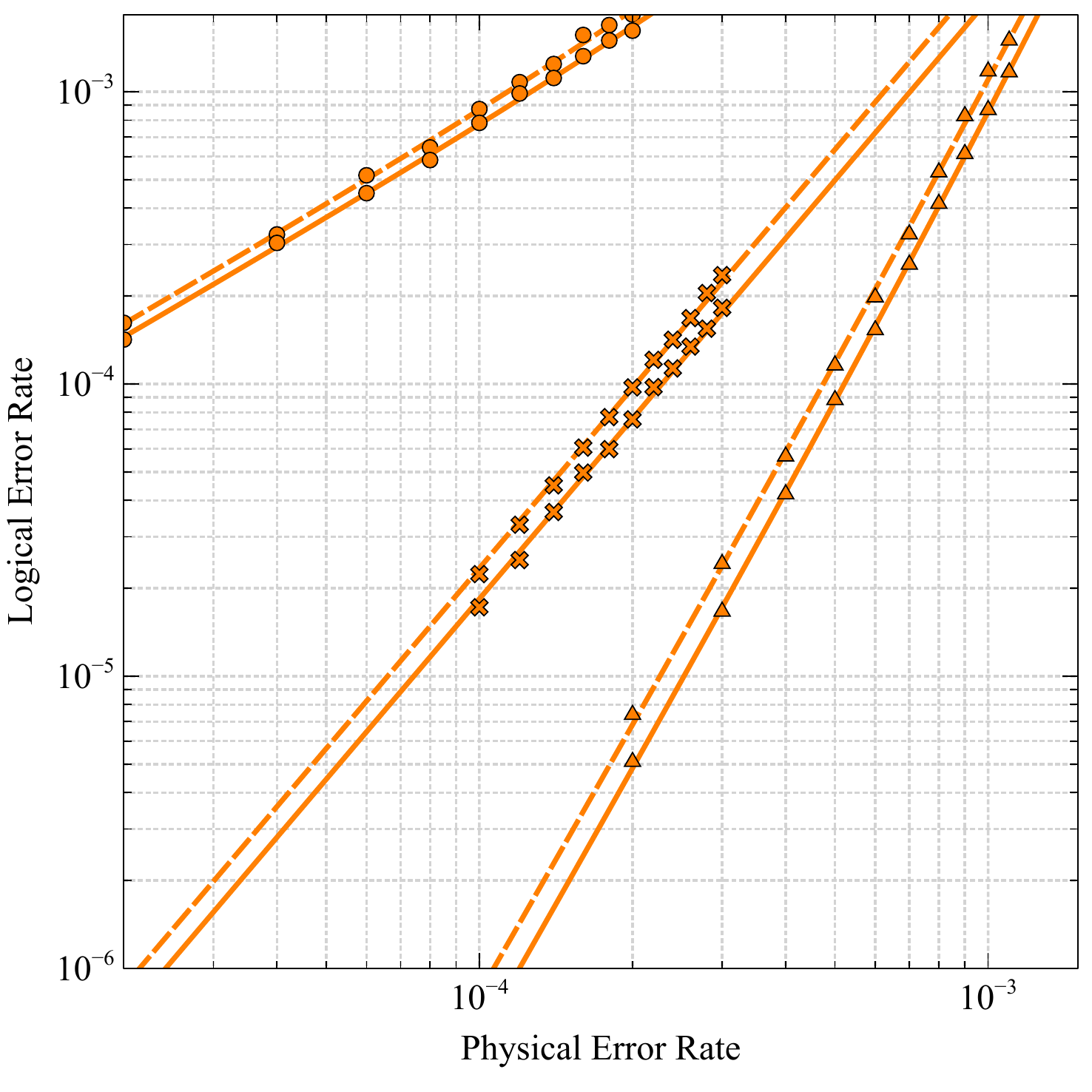}
   \caption{$DP$-Leakage}
   \label{fig:Ng1321} 
\end{subfigure}

\begin{subfigure}[b]{0.5\textwidth}
   \includegraphics[width=0.8\linewidth]{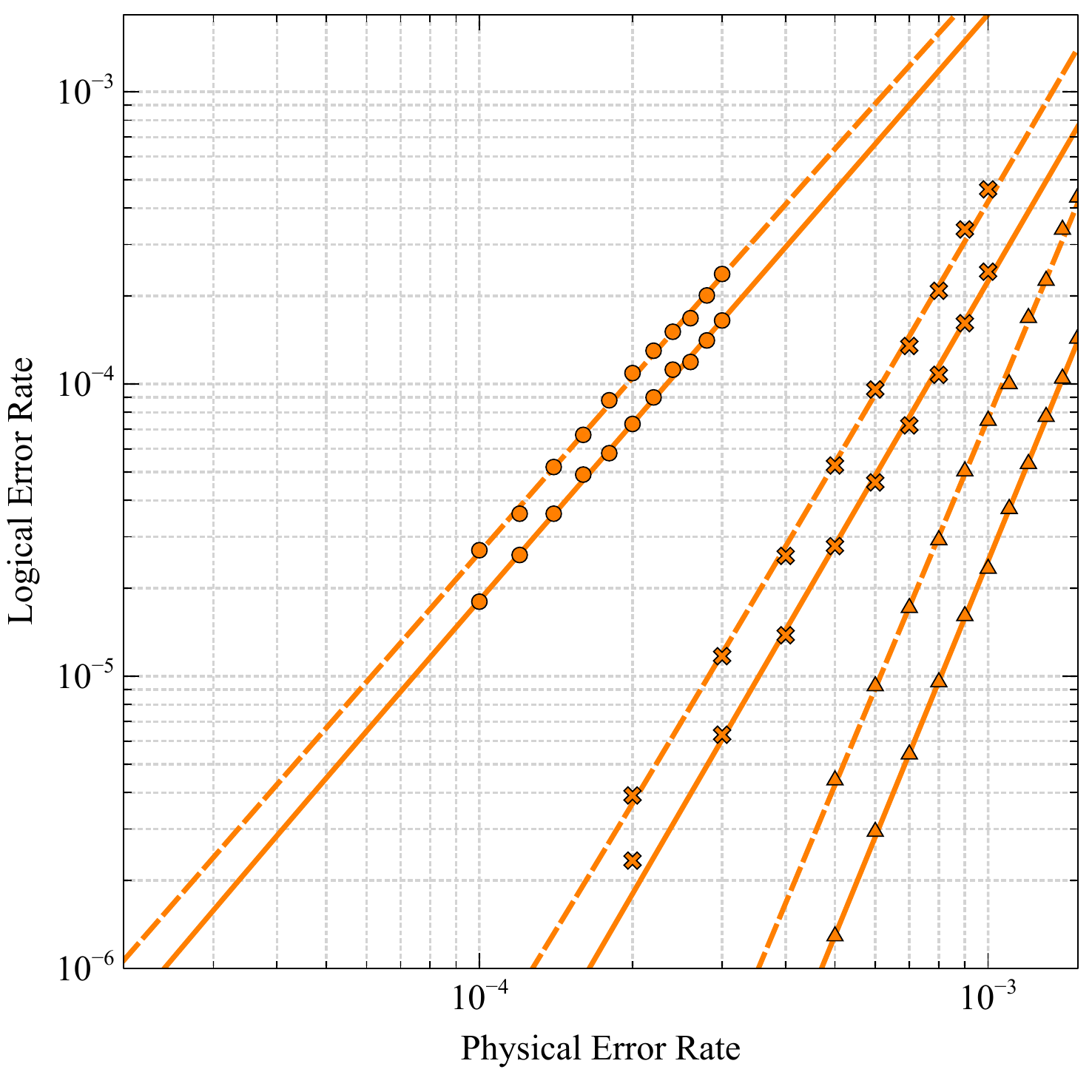}
   \caption{$MS$-Leakage}
   \label{fig:Ng2125}
\end{subfigure}
\caption{A comparison of $DP$-leakage (top) and $MS$-leakage (bottom) logical error rates for the standard surface code using both syndrome-$LR$ (dotted) and swap-$LR$ (solid) at distances $3,5$ and $7$.  As expected, we observe nearly identical scaling, and swap-$LR$ even tends to perform better. Longer-lived leakage errors do not appear to be much more damaging in this regime, as the $MS$-leakage logical error rates are correctly suppressed. However, as these are subspace codes, $DP$-leakage reduces the effective distance.}
\label{swap_numbers}
\end{figure}

\section{Leakage Simulations}\label{Simulations}

In the last section, we established that subsystem surface codes require less relative overhead than the subspace surface codes to realize their full effective code distance in the presence of leakage.  Next, we directly compare the leakage performance of subspace and subsystem surface codes in the low-error regime.  To do so, we perform Monte Carlo simulations of each code in the gate error model.  Although the simplicity of syndrome-$LR$ makes it straightforward to analyze, it does not minimize circuit-volume overhead.  Consequently, we consider the least expensive leakage reduction strategy: swap-$LR$, which requires no qubit overhead, and importantly, may preserve the locality requirements of the qubit lattice.

We focus on error rates in the $10^{-5}$ to $10^{-3}$ range.  Leakage rates spanning this range have been reported for several different qubit architectures \cite{ghosh2015leakage, cerfontaine2016feedback, cerfontaine2019feedback}.  However, the particular choice to center around an error-rate of $10^{-4}$ is motivated by spontaneous Raman scattering rates in ion traps.  It has been estimated that a $200-500$ $\mu s$ gate experiences a spontaneous scattering event with probability $10^{-4} - 10^{-3}$ \cite{yukai2018leakage, brown2018comparing}, and each spontaneous scattering event can populate any state with approximately equal probability.  For qubits based on clock transitions, this splits between the two computational states, as well as two leakage states formed by Zeeman splitting.  We then approximate that a leakage event occurs with probability $\approx 5.0 \times 10^{-5}$ \cite{brown2018comparing}.

\subsection{Surface Codes}
In the depolarizing error model, the extra overhead incurred by introducing additional gauge degrees of freedom degrades the performance of the subsystem surface code.  Even well below threshold, this manifests as higher logical error rates; see Figure \ref{Pauli_errors}.

\begin{figure}[htb!]
\centering
\begin{subfigure}[b]{0.5\textwidth}
   \includegraphics[width=0.8\linewidth]{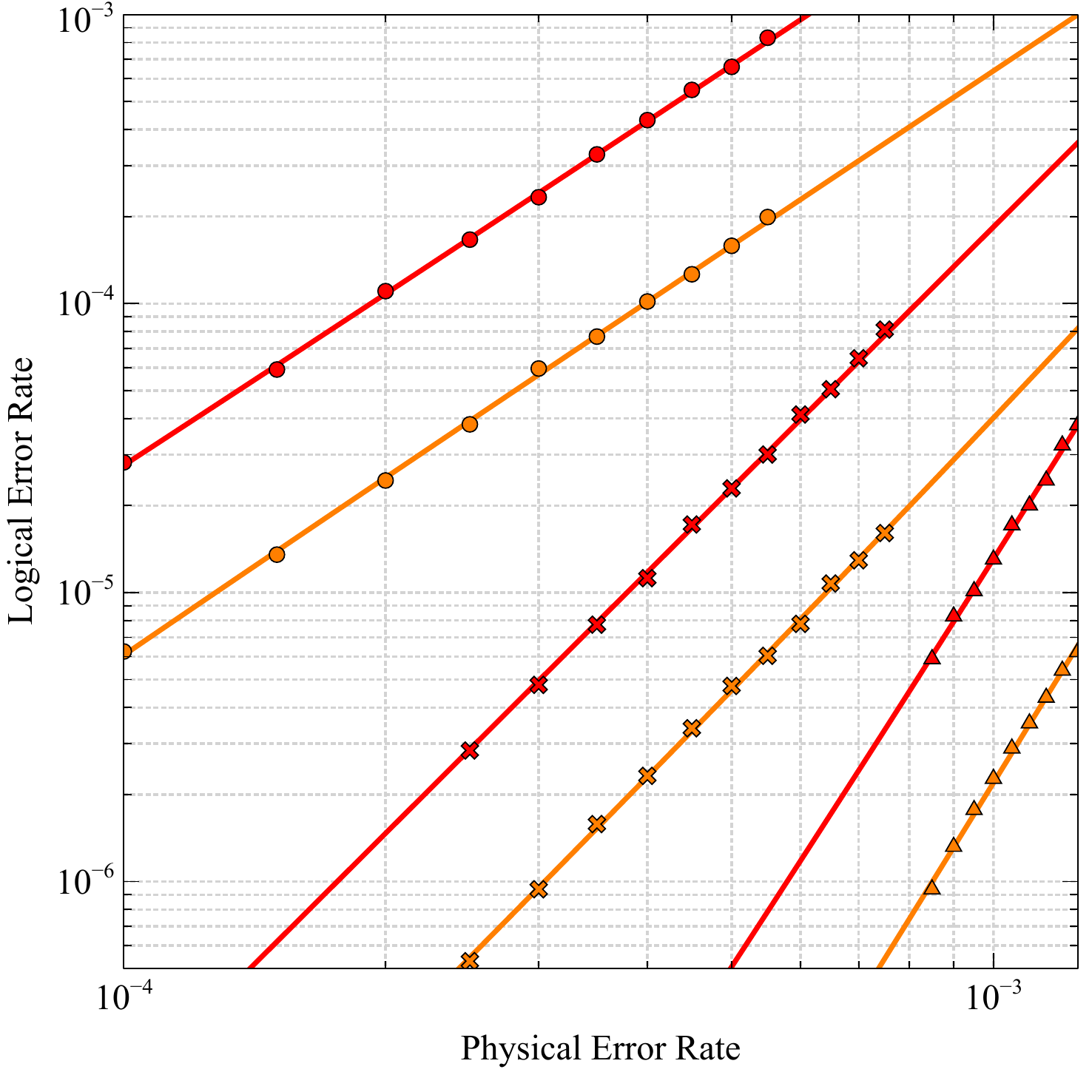}
   \caption{Standard lattice}
   \label{fig:Ng16341} 
\end{subfigure}

\begin{subfigure}[b]{0.5\textwidth}
   \includegraphics[width=0.8\linewidth]{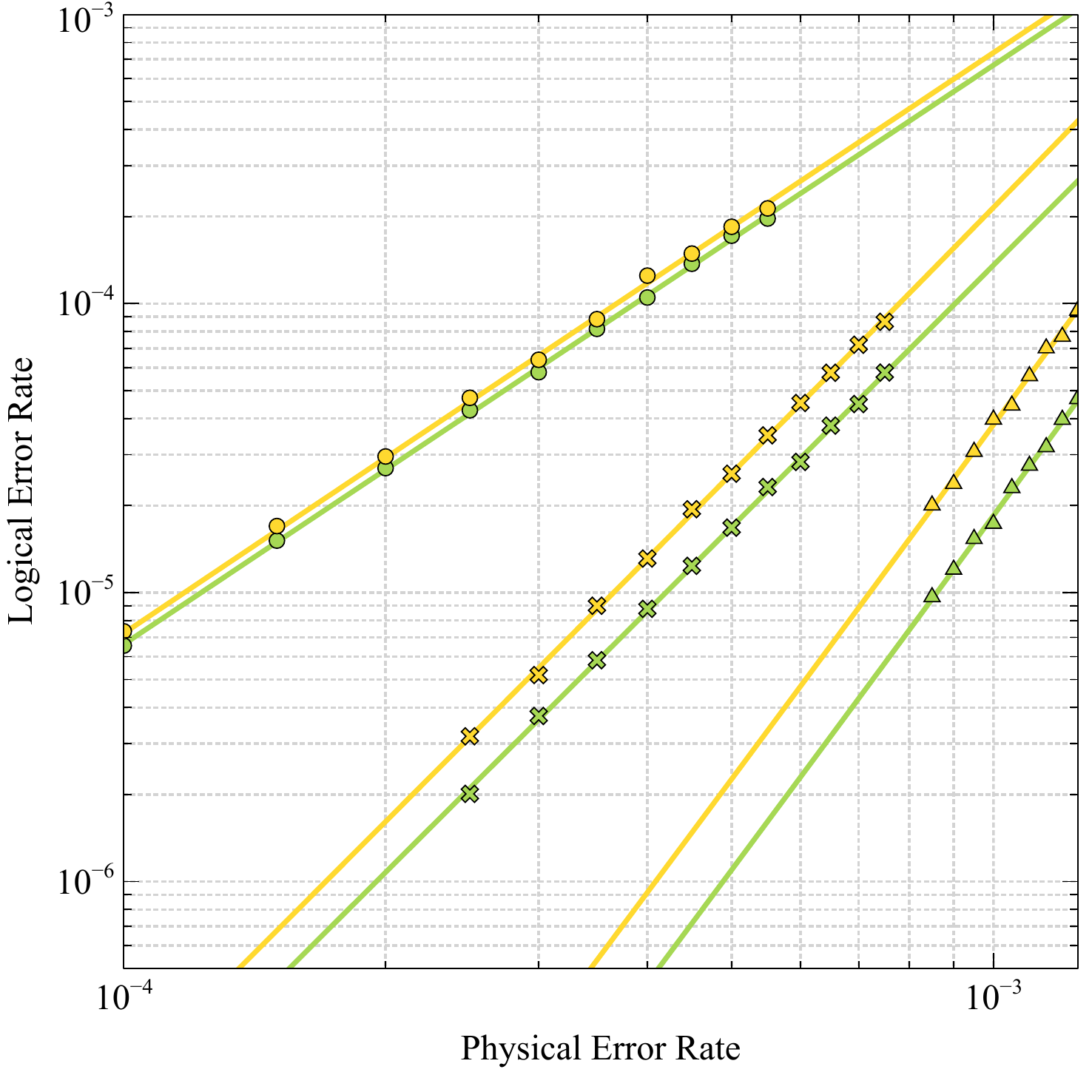}
   \caption{Rotated lattice}
   \label{fig:Ng2651}
\end{subfigure}
\caption{Sub-threshold error rates at distances $3,5,$ and $7$ in the presence of depolarizing noise in both a standard (top) and rotated (bottom) lattice geometry.  As expected, subspace surface codes (green, orange) consistently outperform subsystem surface codes (yellow, red), in some cases by nearly an order of magnitude within the error regime we consider.}
\label{Pauli_errors}
\end{figure}

However, in the presence of leakage, the locality of the subsystem codes allow them to outperform their subspace counterparts at sufficiently low error rates. In order to probe this low-error regime, we restrict ourselves to low-distance codes; see Figure \ref{main:sim}.  We observe good agreement with the expected performance of each code.  Each fit takes the form $p_L \sim p^{d_\text{emp}}$, where $d_{\text{emp}}$ is chosen to minimize the $\chi^2$-distance.  For distances $3$ and $5$, we observe that $d_e - 0.25 < d_{emp} < d_e + 0.5$, where $d_e$ is the expected scaling based on syndrome-$LR$ (see Figure \ref{Overheads}).  For distance $7$, the samples were drawn predominantly from error rates $> 3\times 10^{-4}$.  At higher rates, the distance-$4$ suppression of the Pauli errors are a significant factor when compared to the distance-$2$ suppression of leakage errors. Nonetheless, we observe the expected pronounced distance reduction, with $2< d_{\text{emp}}<3$ for those codes susceptible to leakage, and $d_{\text{emp}} > 4$ for those robust to leakage.  

Most importantly, for all codes that are leakage robust, we observe that $d_{\text{emp}}>d_{e} - 0.25$, giving evidence that the longer-lived leakage errors introduced during swap-$LR$ are significantly less damaging than the space-correlated errors shared with syndrome-$LR$.  Furthermore, we expect that error rates would improve to $d_{\text{emp}}\approx d_{e}$ in this error regime given a decoder that could better handle long-range correlations, such as those based on renormalization groups \cite{duclos2010fast,duclos2010renormalization}.  It is likely that these results would smooth out at higher distances, as off-by-one errors have a significant effect on the low distance scaling.

\begin{figure}[htb!]
\centering
\begin{subfigure}[b]{0.5\textwidth}
   \includegraphics[width=0.8\linewidth]{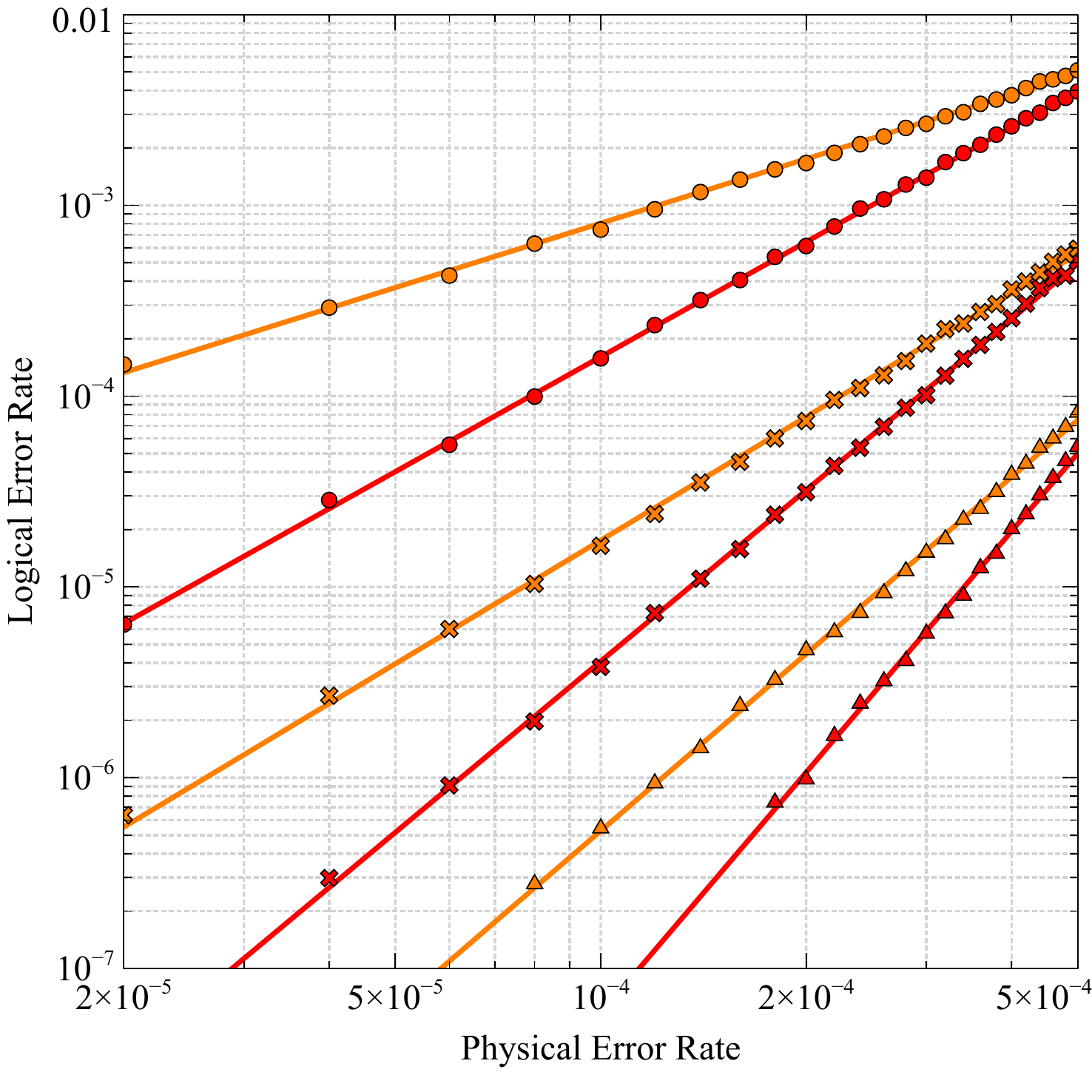}
   \caption{$DP$-leakage}
   \label{fig:Ng163dgs41} 
\end{subfigure}

\begin{subfigure}[b]{0.5\textwidth}
   \includegraphics[width=0.8\linewidth]{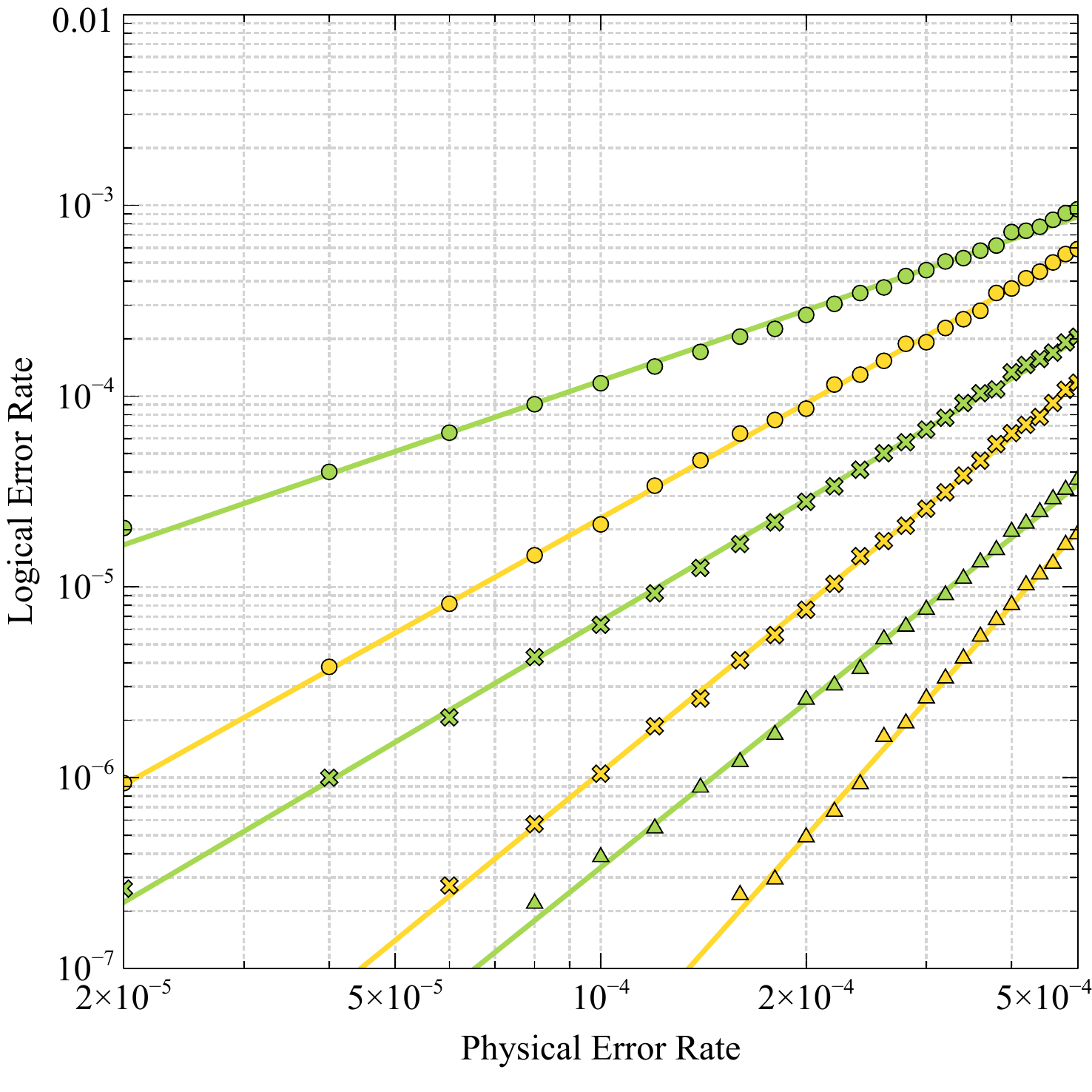}
   \caption{$MS$-leakage}
   \label{fig:Ng26hdhd51}
\end{subfigure}
\caption{Logical error rates at distances $3,5,$ and $7$ in the presence of leakage.  Comparisons are between, in ascending order of qubit overhead, the rotated subspace (green), rotated subsystem (yellow), standard subspace (orange), and standard subsystem (red) surface codes using swap-$LR$.  Each data point was recorded after at least 200 failures, with the longest simulations requiring $\approx 10^9$ trials.} \label{main:sim}
\end{figure}

This establishes that certain subsystem codes relatively outperform subspace codes (at low distance), in the sense that same-distance, same-geometry subsystem codes yield better performance around $p \lessapprox 2.0 \times 10^{-3}$ for $MS$-leakage in a rotated geometry and $p \lessapprox 0.75 \times 10^{-3}$ for $DP$-leakage in a standard geometry.  This contrasts with a local depolarizing model, in which subsystem codes do not relatively outperform surface codes in any error regime.

It is vitally important to note that this does not account for the $1.75\times$ qubit overhead required for subsystem codes of the same geometry as subspace codes.  So while the arguments of Section \ref{Robustness} demonstrate that subsystem codes offer better per-qubit distance protection in the presence of leakage, this advantage manifests at much lower error rates.  Although we are unable to probe these error regimes directly, we can give coarse upper bounds using heuristic estimates based on thresholds, see Appendix \ref{threshold_estimates}.  We estimate that, at the very least, $p$ must be $\lessapprox 2.5 \times 10^{-4}$ in the presence of $DP$-leakage and $ \lessapprox 0.32 \times 10^{-4}$ in the presence of $MS$-leakage to potentially see a per-qubit benefit at sufficiently high distance.

\subsection{Bacon-Shor Codes}

We next consider Bacon-Shor codes for use in the $MS$-leakage model.  
Bacon-Shor codes are subsystem codes defined on a lattice of $d \times d$ data qubits.  They have $2(d-1)$ stabilizer generators, with $X$- and $Z$-type generators indexed by $1 \leq k \leq d-1$ and defined by $$X_k := \prod\limits_{i=1}^d X_{i,k}X_{i,k+1} \text{ and } Z_k := \prod\limits_{j=1}^d Z_{k,j}Z_{k+1,j}$$ where $P_{i,j}$ represents the operator $P$ acting on the qubit in the $i$th row and $j$th column of the lattice.  Then the stabilizer group is generated by the set of all $X_k,Z_k$.  As subsystem codes, a generating set for the gauge operators is given by, $$X_{i,j} := X_{i,j}X_{i,j+1} \text{ and } Z_{i,j} := Z_{i,j}Z_{i+1,j},$$ where addition is performed modulo the lattice size $d$.  With this orientation, $X_L$ consists of $X$ operators spanning the north-south boundaries of the lattice, while $Z_L$ consists of $Z$ operators spanning the east-west boundaries of the lattice.  In particular, Bacon-Shor codes share the same efficient data qubit scaling as rotated surface codes, forming a family of $[[d^2,1,d]]$ codes.

Although there is additional ancilla overhead, this gives us two code families to compare more closely, namely the Bacon-Shor and rotated subspace surface codes in an $MS$-leakage model and with bare-ancilla extraction.  The former sacrifices syndrome information to perform localized checks, with weight two gauge operators ensuring that every space-correlated error due to $MS$-leakage has weight one.  The latter sacrifices locality for additional syndrome information, increasing the number of high-weight correctable errors while introducing other damaging correlated errors in the process.  The result of this tradeoff is that, at low distances, Bacon-Shor codes yield a per-qubit error-corrective advantage in the presence of $MS$-leakage at reasonable error rates, ranging from $10^{-4} - 10^{-3}$; see Figure \ref{BaconShorData}.

\begin{figure}[htb!]
\centering
\begin{subfigure}[b]{0.5\textwidth}
   \includegraphics[width=0.8\linewidth]{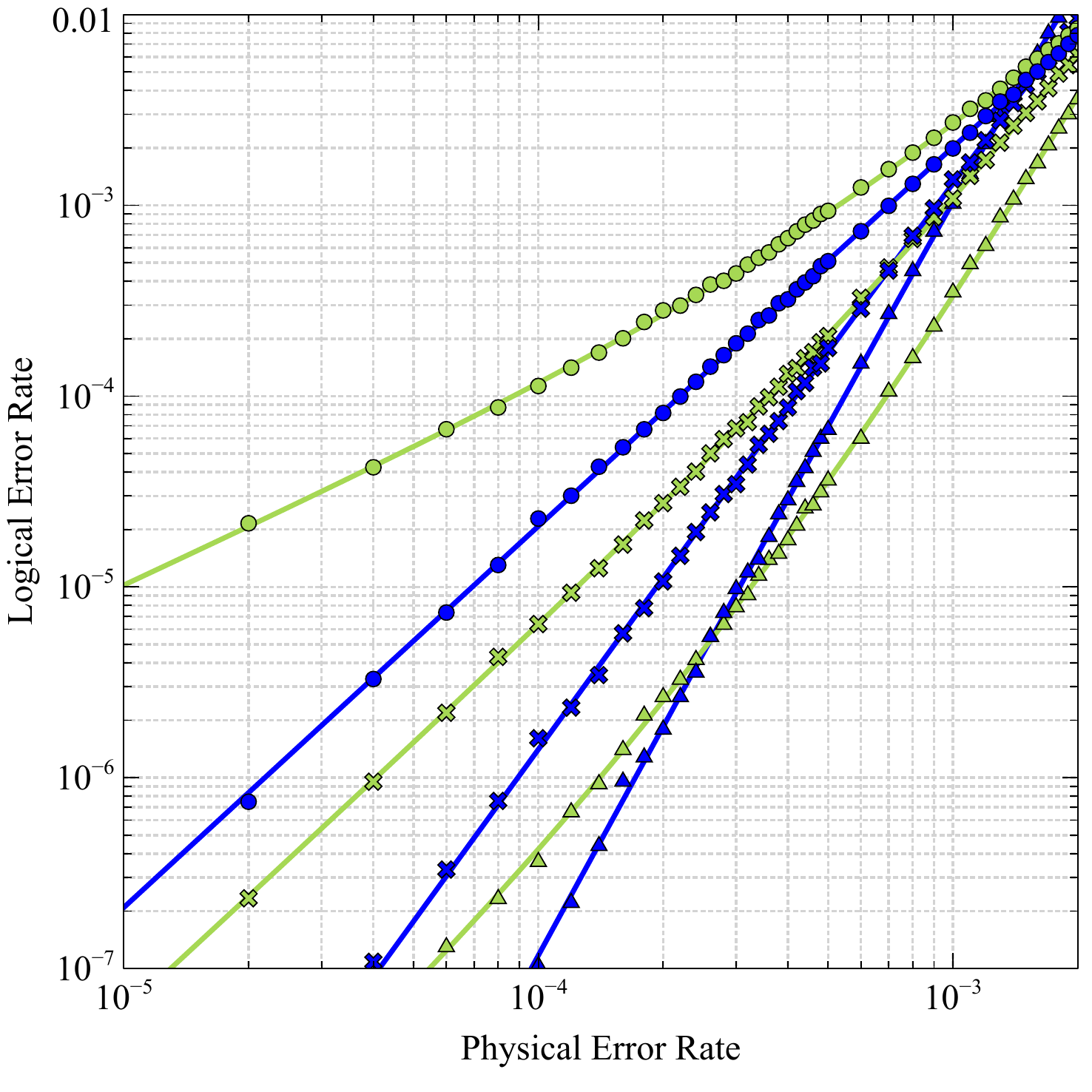}
   \caption{$p_\ell = p_d$}
   \label{fig:Ng1634sdgsdg1} 
\end{subfigure}

\begin{subfigure}[b]{0.5\textwidth}
   \includegraphics[width=0.8\linewidth]{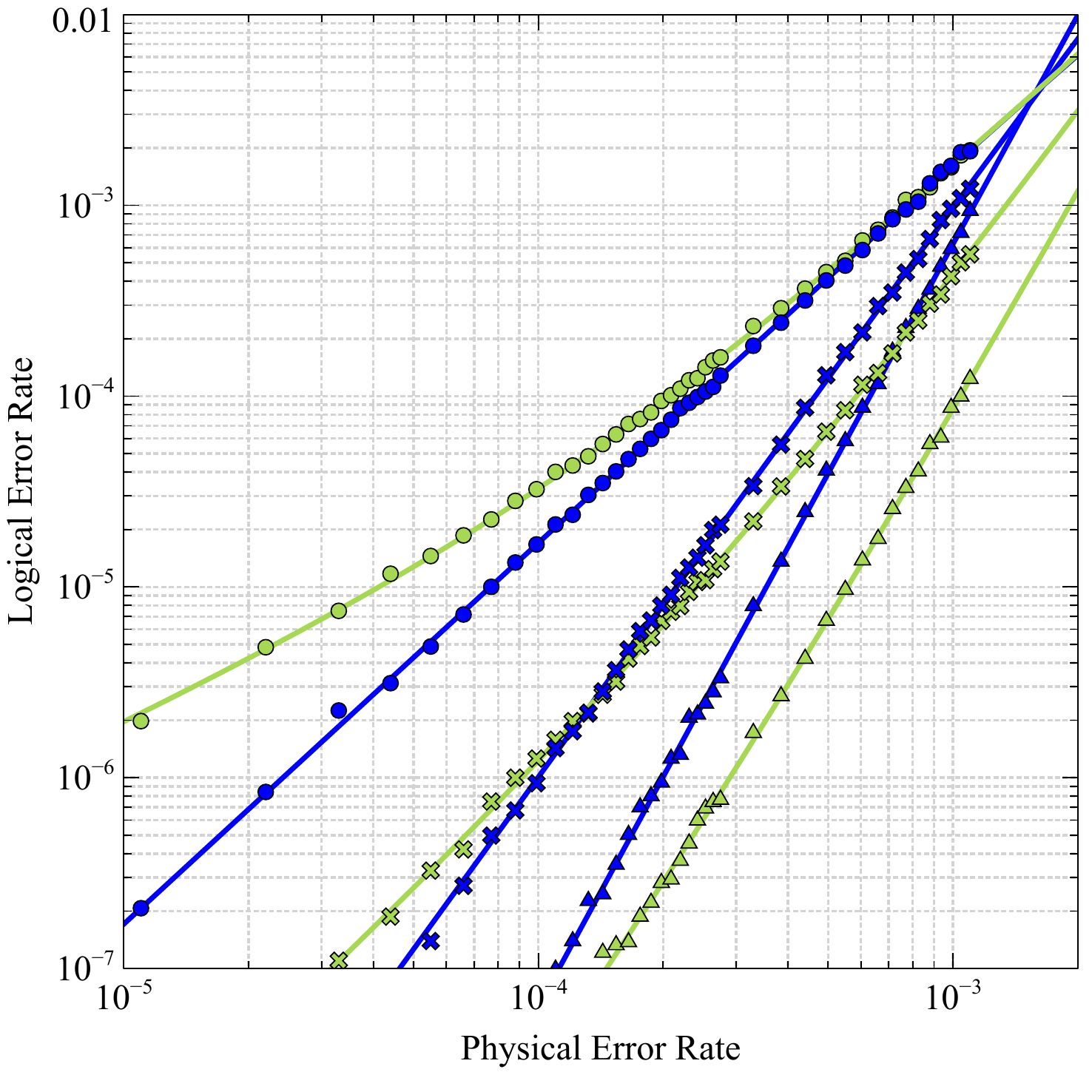}
   \caption{$p_\ell = 0.1p_d$}
   \label{fig:Ng265dhgdsh1}
\end{subfigure}
\caption{A comparison of distance $3,5,$ and $7$ rotated subspace surface codes (green) with Bacon-Shor codes (blue).  Here we consider both $p_{\ell} = p_d$ (top) as well as $p_{\ell} = 0.1 p_d$ (bottom), which may be applicable to different architectures.  In both cases, Bacon-Shor codes begin to demonstrate an advantage at errors rates $\lessapprox 10^{-3}$ at distance-$3$.  Furthermore, this advantage persists within reasonable leakage regimes through distance-$7$.  Each fit is of the form $Ap^{d_e} + Bp^{d_e + 1}$, where $d_e$ is the minimum number of faults that may cause a logical error and $A,B$ are the fitting parameters.}
\label{BaconShorData}
\end{figure}

At low distance, this advantage manifests at error rates as high as $p \lessapprox 1.2 \times 10^{-3}$.  Even when the leakage rate is an order-of-magnitude less than the depolarizing rate, the distance-$3$ Bacon-Shor code begins to outperform the corresponding surface code at error rates as high as $p \lessapprox 5.0 \times 10^{-4}$.  Of course, these gains are more pronounced at higher leakage-to-depolarizing ratios.  At lower error rates, one can observe the damage caused by leakage as the logical error suppression tapers off.  

However, as the number of stabilizer checks in Bacon-Shor codes scale sublinearly in the lattice size, the family does not exhibit a threshold.  Consequently, their advantage cannot persist at higher distances, although they can yield significant error-suppression which may suffice for the desired memory time \cite{Napp:2012, brooks2013fault}.  Coupled with simplified preparation and potential benefits against other correlated noise sources, this gives evidence that Bacon-Shor codes might prove advantageous in near-term fault-tolerance experiments or in low to intermediate distance error protection, within certain noise models \cite{Li:2018}. 

\subsection{Surface Codes with Unverified Cat States}

Thus far, we have focused on codes defined on square qubit lattices with bare-ancilla syndrome extraction.  While this might be preferable architecturally, we may relax this requirement to consider slightly larger ancilla states.  Unlike the preceding results, which can be extended straightforwardly to the presence of correlated leakage by passing to syndrome-$LR$, this relies unavoidably on an independent leakage model.  For some platforms, such as certain quantum dot architectures, these correlated leakage events must be handled.  For others, we can leverage the independence assumption to combine the locality of the Bacon-Shor code with the performance of the surface code by encoding the gauge-fix directly into the ancilla.  This amounts to performing the simplest Shor-style extraction, without verification or decoding.

Typical mechanisms for fault-tolerance involve preparing and verifying entangled ancilla states \cite{shor1996fault,divincenzo2007effective,steane1997active,knill2005scalable} or the use of flags \cite{chao2018quantum, chamberland2018flag}.  Unfortunately, as a leakage event can manifest many Pauli errors, it can often fool such verification schemes.  Although this can make typical fault-tolerant circuits more difficult to engineer \cite{fortescue2014fault}, the circuit locality of Shor-style syndrome extraction may provide a direct benefit when leakage events occur independently; see Figure \ref{cat}.

\begin{figure}[htb!]
\includegraphics[width=\linewidth]{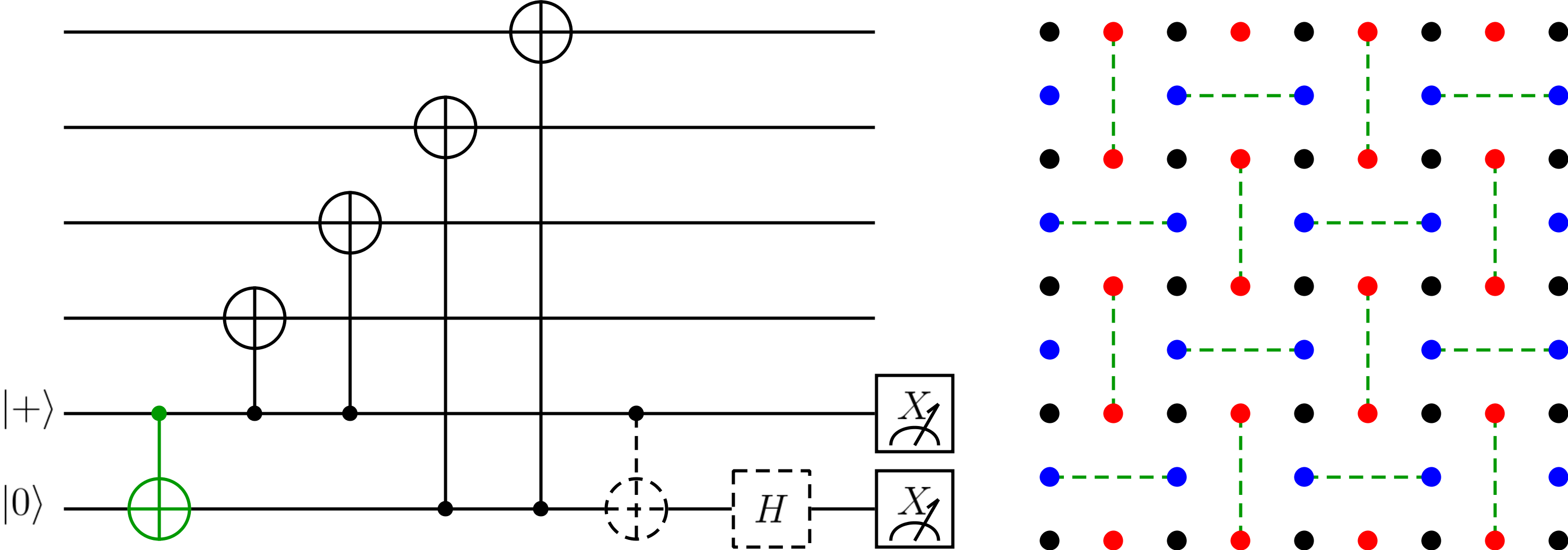}
\caption{Syndrome extraction on the rotated surface code with unverified two-qubit cat states.  The usual circuit involves decoding the ancilla \cite{Yoder:2017b}, with dashed gates omitted as verification can be fooled by leakage.  On the right is a distance-$5$ code with black dots representing data qubits, red dots representing $X$-type ancilla qubits, and blue dots representing $Z$-type ancilla qubits.  Each ancilla extracts the syndrome from its two nearest neighbor data qubits, with green connections representing longer-range ancilla entangling operations.}
\label{cat}
\end{figure}

Note that without entangling the ancilla, each measurement simply corresponds to measuring the gauge generators of the Bacon-Shor code.  In this sense, entangling the ancilla encodes a gauge-fix of the Bacon-Shor code to the surface code \cite{li20182}, while robustness to $MS$-leakage is unaffected by this leading operation.  One must at least pass to four-qubit cat states in the presence of $DP$-leakage, as any pair of qubits in the stabilizer supports two qubits of a minimal distance logical operator.

\begin{figure}[htb!]
\centering
\begin{subfigure}[b]{0.5\textwidth}
   \includegraphics[width=0.75\linewidth, height = 5.88cm]{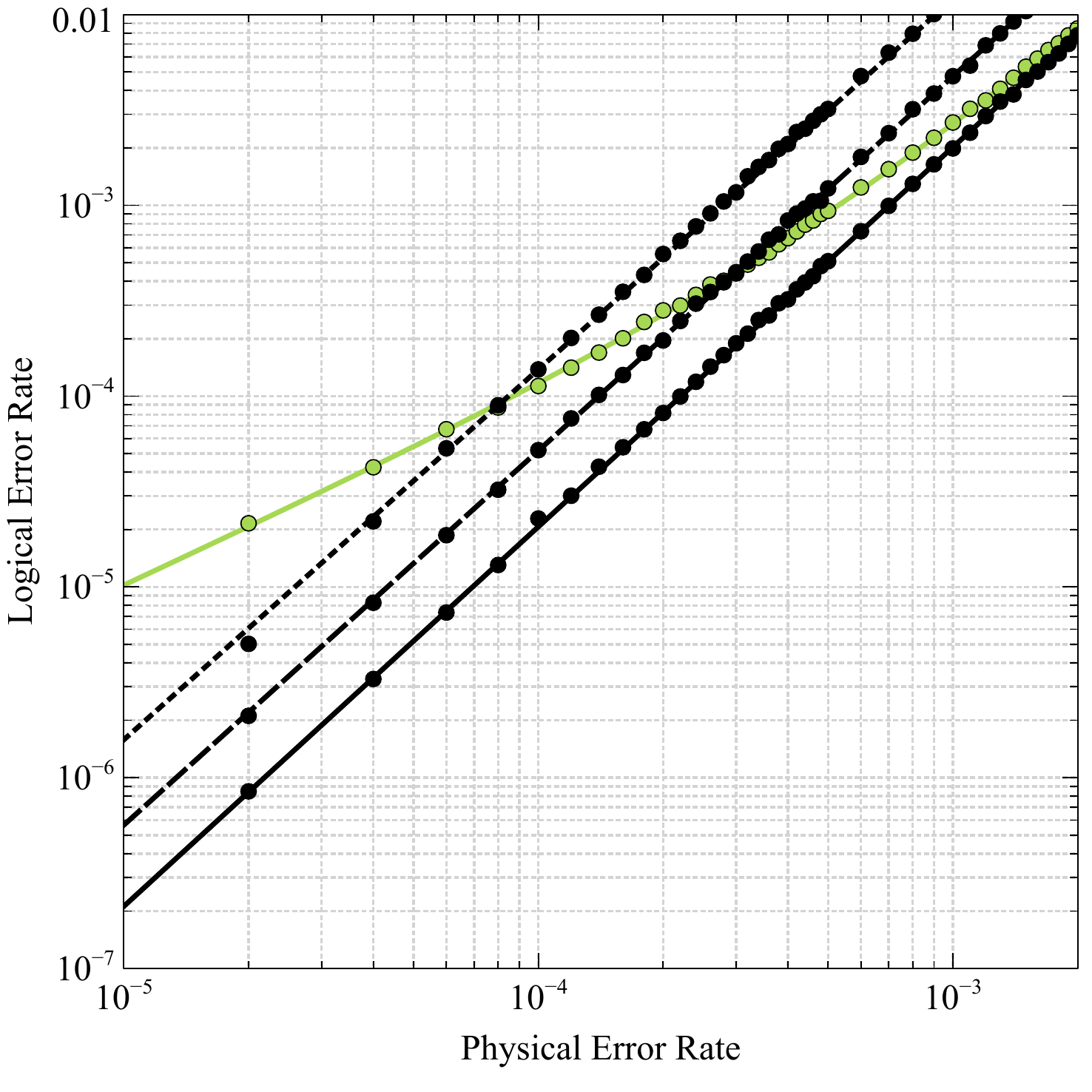}
   \caption{Distance-$3$}
   \label{fig:Ng1634sdgssdfsdg1} 
\end{subfigure}

\begin{subfigure}[b]{0.5\textwidth}
   \includegraphics[width=0.75\linewidth, height = 5.88cm]{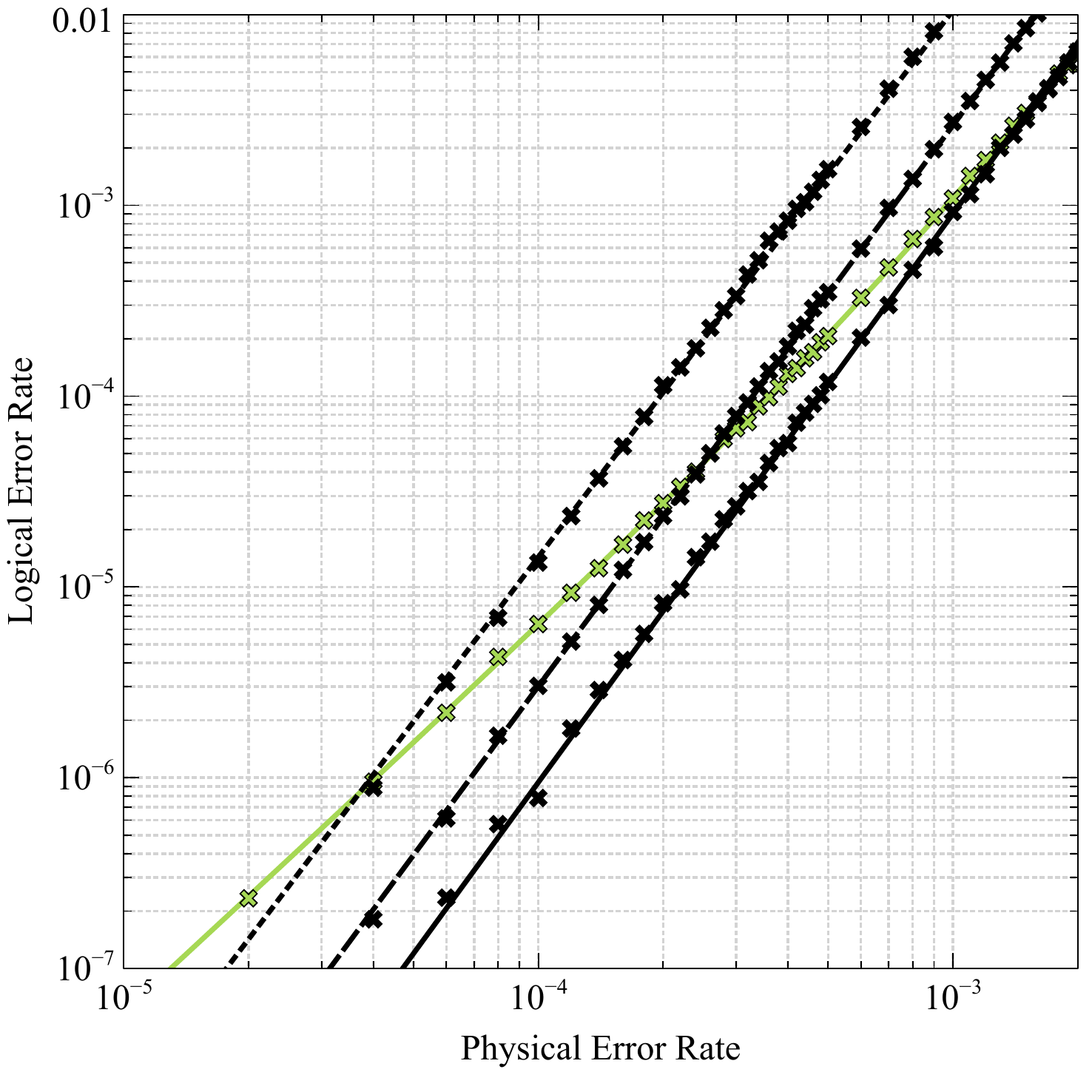}
   \caption{Distance-$5$}
   \label{fig:Ng265dhgasdfsdsh1}
\end{subfigure}

\begin{subfigure}[b]{0.5\textwidth}
   \includegraphics[width=0.75\linewidth, height = 5.88cm]{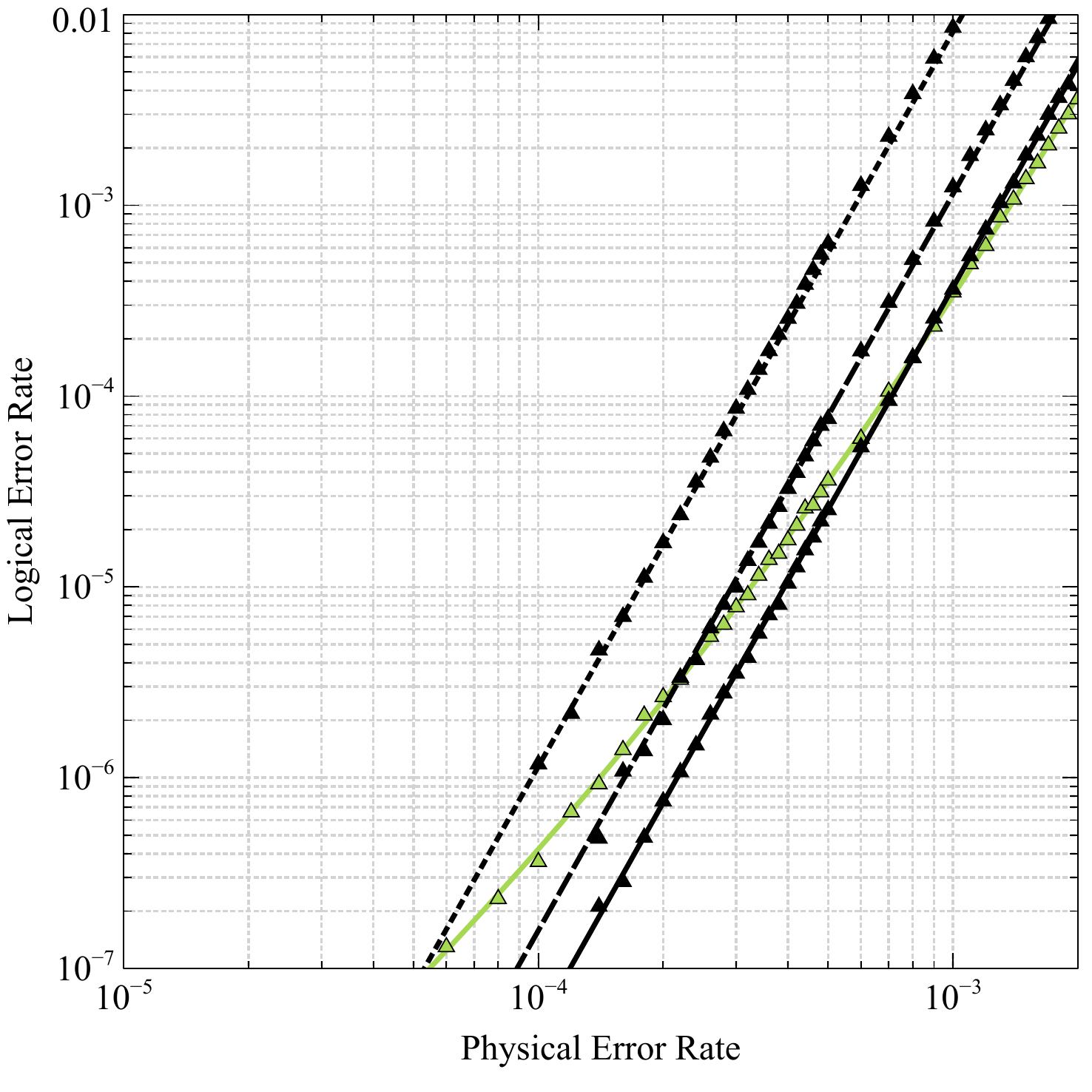}
   \caption{Distance-$7$}
   \label{fig:Ng265asdfsdhgdsh1}
\end{subfigure}
\caption{Error-rate comparison in the presence of $MS$-leakage on the rotated surface code.  Green represents swap-$LR$, the dotted black line represents the gate-$LR$ proposed in \cite{suchara2015leakage}, and the dashed black line represents intermediate-$LR$ (see Appendix \ref{auxiliary}).  The solid black line represents swap-$LR$ combined with unverified two-qubit cat state extraction.}
\label{cat_data}
\end{figure}

Compared to the previous fault-tolerance approach of \cite{suchara2015leakage}, in which leakage reduction is performed after each gate, we observe an orders-of-magnitude improvement in performance; see Figure \ref{cat_data}.  Similar to the Bacon-Shor code, using two-qubit cat state extraction outperforms bare-ancilla extraction below error rates $\approx 10^{-3}$.  However, as this family exhibits a threshold, we would expect that this advantage is better preserved at higher distances.  When compared directly to Bacon-Shor codes, two-qubit cat state extraction of the surface code performs marginally worse at distance-$3$, and then begins to outperform.

However, in the presence of correlated leakage events, the preparation of an entangled ancilla state may be badly corrupted.  In this case, two-qubit cat states do not preserve the effective distance, highlighting that the more damaging the error model, the more valuable the subsystem locality.  The resulting comparison between cat state extraction of the surface code and bare-ancilla extraction of the Bacon-Shor code in the presence correlated leakage events roughly mirrors the comparison between bare-ancilla extraction of both codes with independent leakage events; see Appendix \ref{correlated}.

\section{Conclusions}\label{Discussion}

In this work, we have highlighted an intrinsic error-corrective benefit of subsystem codes in the presence of leakage.  This is because leakage naturally manifests uncontrolled `hook' errors, whose damage is limited by local measurements.  We have quantified several crossover points below which these correlated errors are more damaging than the entropic effects that are well-handled by the subspace surface code.  These error rates range between approximately $10^{-5}-10^{-3}$ at the low distances we have considered, and may be relevant to limiting error models in certain architectures, such as spontaneous scattering in ion traps.

It is important to note that these advantages are modest.  For the geometrically local codes we have considered, subsystem error-correction only provides a benefit subject to certain code constraints or particularly damaging leakage models.  Even here, these benefits can only be observed at low error-rates that are only just being broached by current technologies \cite{ballance2016high, cerfontaine2019feedback, barends2014superconducting}.  However, there are certain architectures that could have even more harmful leakage dynamics or even higher leakage rates \cite{andrews2018,negnevitsky2018repeated}; in such cases, we would expect the subsystem code advantage to increase.

More generally, this work illustrates a potential application for subsystem codes which do not yield an immediate error-corrective advantage in independent depolarizing models.  As a generalization of subspace codes, subsystem codes offer a broad design space, with a variety of interesting constructions including Bravyi-Bacon-Shor codes \cite{bravyi2011subsystem, Yoder:2019}, the five-squares code \cite{suchara2011constructions}, subsystem color codes \cite{Bombin:2010b}, and others \cite{aly2008subsystem, sarvepalli2012topological, gayatri2018decoding, marvian2017error}.  Leakage is a natural error-model in which the subsystem structure of these codes offers a direct error-corrective benefit when compared with any gauge-fix, past architectural and parallelizability considerations.  We would expect similar advantages in other non-Markovian error models, where the relevant baths are local and determined by the qubit connectivity \cite{terhal2005fault}.

For more immediate practical applications, we have given evidence that leakage can be less damaging in certain architectures.  This occurs in a non-interactive model of leakage, and can be applied to M{\o}lmer-S{\o}rensen gates in ion-traps. In practice, one should leverage all of the specificities of their architecture to handle leakage events optimally. For systems that suffer from very damaging leakage, we have shown that Bacon-Shor or subsystem surface codes may be preferable, particularly in near-term fault-tolerance experiments at low distances. Finally, we have shown that in certain leakage models, using unverified two-qubit cat state extraction on the surface code may be preferable at error-rates $\lessapprox 10^{-3}$ at the same-distance.  This yields an order-of-magnitude improvement over existing fault-tolerant proposals for gate-by-gate leakage reduction \cite{suchara2015leakage}.

There are several extensions to this work worth consideration.  Most immediately, each of these codes and leakage reduction strategies should be compared on a qubit-to-qubit basis, given the constraints of a particular architecture and leakage model.  To do this, we would have to probe the low-error regime of higher-distance codes.  Generalized path-counting \cite{beverland2018role, fowler2012surface} or splitting methods \cite{bravyi2013simulation} may be better suited to the task than Monte Carlo simulations.  Ultimately, we are interested in performing very long computations, and so understanding code behavior at high distances and low error-rates is key.

Additionally, one could expand the search for leakage resilient codes past geometrically local codes, which are bounded asymptotically by $d = O(\sqrt{n})$, and consider more general low-density parity check codes \cite{bravyi2011subsystem}.  Even among geometrically local codes, one could consider more qubit-efficient encodings, such as the $4.8.8$ color codes \cite{landahl2011fault} or triangle codes \cite{Yoder:2017b}.   One could also try to generalize fault-tolerant syndrome extraction to account for leakage events \cite{fortescue2014fault}.  However, this will also introduce additional overhead, while preparation and verification of the required resource states will prove more difficult depending on the leakage model. Along the same line, one could investigate concatenated codes for non-local architectures.  However, leakage will further exacerbate the low thresholds inherent to concatenated codes.

One could also look for improvements using more advanced decoders or leakage detection \cite{suchara2015leakage}.  We have restricted our attention to minimum-weight perfect matching, which only captures edge-like correlated errors in the decoder graph.  Renormalization group \cite{duclos2010fast,duclos2010renormalization} or machine-learning methods \cite{maskara2018advantages,chamberland2018deep} could better handle higher-weight correlated error patterns.  These decoders might also utilize the extra information that is gained from auxiliary qubit measurement, which is wasted in a simple minimum-weight perfect matching scheme, and could lead to improved thresholds.  

Finally, one might consider subsystem codes in the presence of other correlated noise models \cite{fowler2014quantifying, terhal2005fault}, particularly those that are determined by the lattice geometry or gate connectivity, such as crosstalk. In these cases, local subsystem error-correction may provide benefits for handling correlated errors that lie outside the scope of independent depolarizing noise.

\section{Acknowledgements}
The authors thank Andrew Cross and Martin Suchara for providing the toric code simulator from which the other simulators were built, with permission from IBM.  They additionally thank Pavithran Iyer, Aleksander Kubica, Yukai Wu, and Ted Yoder for helpful comments and discussions.  Finally, they acknowledge the Duke Computing Cluster for providing the computational resources for simulations.  This research was supported in part by NSF (1717523), ODNI/IARPA LogiQ program (W911NF-10-1-0231), ARO MURI (W911NF-16-1-0349), and EPiQC - an NSF Expedition in Computing (1730104). 
\bibliography{bibliography}
\appendix
\clearpage

\section{Data Leakage}\label{myleakages}

In this section, we detail the requirements for a gate scheduling to ensure that data leakage in the $DP$-model does not damage the effective distance of subsystem surface codes.  Unfortunately, gate schedulings for subsystem codes are more complicated than for subspace codes.  This is because measuring anticommuting gauge operators may randomize their syndrome outcomes.  Consequently, interleaving and parallelizing syndrome extraction must be performed carefully.  For some architectures like ion-traps, errors are dominated by gate application, with memory errors occurring orders-of-magnitude less frequently \cite{trout2018simulating}.  In these cases, parallelization becomes relatively less important.  However, for other architectures like superconductors with gates limited by memory error, this parallelization is vital \cite{fowler2012towards}.

\subsection{Syndrome-$LR$}
There are two properties that we want our syndrome extraction to satisfy.  The first is that it is \emph{correct}: its outcomes correctly reflect Pauli errors on the data in the ideal case (i.e., there is no randomization due to measuring anticommuting operators consecutively).  The second property we require is that two triangular gauge operators of the same type overlapping on any qubit cannot both interact with that qubit in the second time-step after their initialization.  This requirement removes bad correlated errors due to data leakage.  We call a syndrome extraction scheduling satisfying this property \emph{good}.

In order to avoid the pitfall of measuring anticommuting operators, the simplest syndrome extraction is to first measure all $X$-type gauges in parallel, then measure all $Z$-type gauges in parallel.  In \cite{bravyi2013subsystem}, a rolling syndrome extraction was defined that ensured that no qubit was idle at any time-step during the syndrome extraction of the subsystem surface code.  To simplify the discussion, we perform the simplest syndrome extraction in our gate error model without idling errors. We further motivate this simplification by detailing a fully-parallelized (on data qubits) good and correct syndrome extraction scheme for subsystem surface codes using syndrome-$LR$ in Appendix \ref{scheduling}.

It remains to argue that data leakage may not cause damaging correlated errors in a good and correct syndrome extraction scheme. Note that if an ancilla interacts with a data leakage in the third time-step after initialization, then it interacts with no other data qubits after that point and cannot spread a correlated error.  If an ancilla interacts with a data leakage in the first time-step after initialization, then the resulting correlated error is equivalent to an error on the leaked qubit, up to gauge transformation.  Thus, two-qubit errors due to data leakage may only occur in triangular gauge operators which both interact with the leakage in the second time-step after initialization.
 
 Furthermore, each gauge operator may only propagate Pauli errors of the same type.  As error-correction is performed independently, it suffices for leakage robustness that no two gauge operators of the same type both introduce correlated errors due to a single leakage.  Thus, goodness of syndrome extraction ensures that any error of either $X$- or $Z$-type caused by a single data leakage is contained in one triangular gauge operator of that type.  As this falls into the ancilla leakage model, we may conclude that the effective distance of the code is preserved.  See Figure \ref{bad_data_leakage} for the damaging data leakage that might occur in a bad syndrome extraction scheduling.  This completes the analysis for subsystem surface codes using syndrome-$LR$.
 
 \subsection{Swap-$LR$}
 
Similar to data leakage using syndrome-$LR$, we must ensure that errors of either $X$- or $Z$-type produced by a single data leakage event are contained in the support of a single gauge operator.  In the case of syndrome-$LR$, long-range correlated errors could occur when two triangular gauge operators of the same type both interact with a data leakage in time-step two.  In the case of swap-$LR$, because an ancilla leakage swaps places with a data qubit in time-step three, long range correlated errors could occur when two triangular gauge operators interact with a data qubit in any combination of time-steps two and three.

Fortunately, Figure \ref{swap_scheduling} demonstrates a gate scheduling that avoids these distance-reducing correlated errors.  It includes a single long-range correlated error that is not too damaging, occurring when a data qubit occupying the short-side of a hexagonal stabilizer leaks; see the left-hand side of Figure \ref{bad_data_leakage}.
 
 \begin{figure}[htb!]
\includegraphics[width=\linewidth]{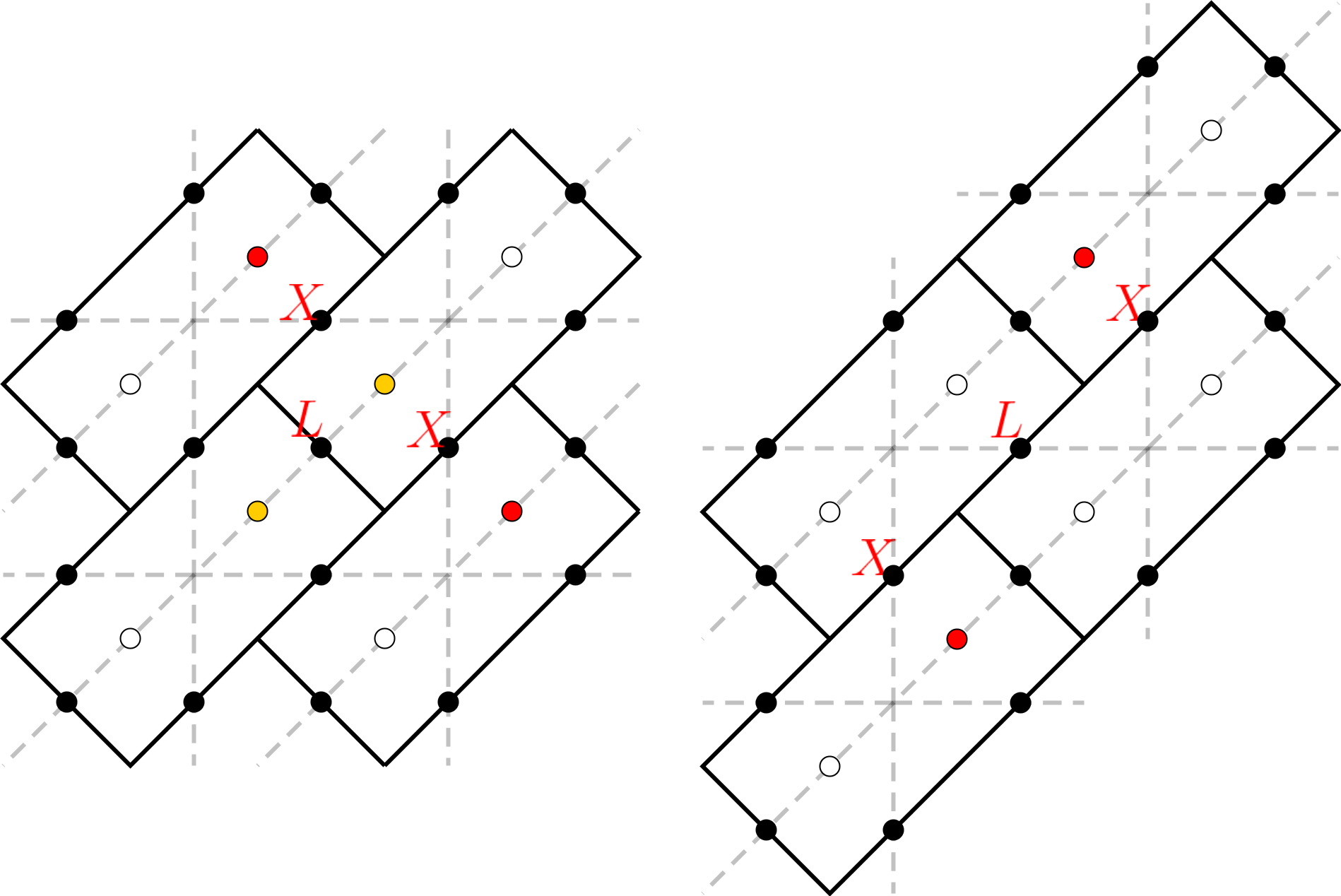}
\caption{Potential data leakage errors.  The left-hand error is equivalent to an ancilla leakage, producing the red excitations if it depolarizes to $I$ and the red plus yellow excitations if it depolarizes to $X$, and so it cannot reduce the effective distance.  In contrast, if the right-hand leakage error depolarizes to the identity, it causes a separated pair of excitations, reducing the distance. With good syndrome extraction, only one of the $X$-errors on the right-hand side (caused by propagation of an $X$-error to an ancilla) may occur.}
\label{bad_data_leakage}
\end{figure}

\section{Gate Schedulings}\label{scheduling}

\subsubsection{Simulation Details}

We follow the simulation technique in \cite{suchara2015leakage}: for a code of distance-$d$, we simulate $d$ time-steps of faulty syndrome extraction.  Then, we run an additional round of perfect syndrome extraction, converting any remaining leakage errors to depolarizing errors in the process.  As we consider the simulated round of error-correction as occurring during a longer computation, we associate a leading probability of leakage according to the lifetime of each data qubit at its initialization.  Each simulation ran until at least $200$ failures were recorded, and for those simulations that required less than $5 \times 10^{8}$ trials, until at least $1000$ failures were recorded.

For all simulations, we use the standard minimum-weight perfect matching decoder.  The edge weights are generated by iterating over all single Pauli faults, and associating a relative probability $p$ to each edge of the decoder graph.  The weight of the edge is then defined to be $-\ln(p)$.  Boundary conditions are handled according to the techniques of \cite{Wang:2009}.  

Note that we do not iterate over leakage events as they may cause fault configurations that result in hyperedges.  Furthermore, for codes that experience a distance reduction in the presence of leakage, we would have to insert edges lying outside a fundamental cell.  Thus, these edge weights are suboptimal.  However, to counteract this, we add edges corresponding to known single-time-step leakage fault configurations occurring with probabilities linear in $p$.  We heuristically weight these as four times the maximum edge probability generated by Pauli configurations.

\subsubsection{Timings}
For the subspace surface code, we use different but standard gate schedulings depending on the lattice geometry, see Figure \ref{subspace_schedulings}.  Note that the decoder graph is always populated by more edges than in the independent depolarizing model, corresponding to the additional correlated errors that may be caused by leakage.

 \begin{figure}[htb!]
\includegraphics[width=.5\linewidth]{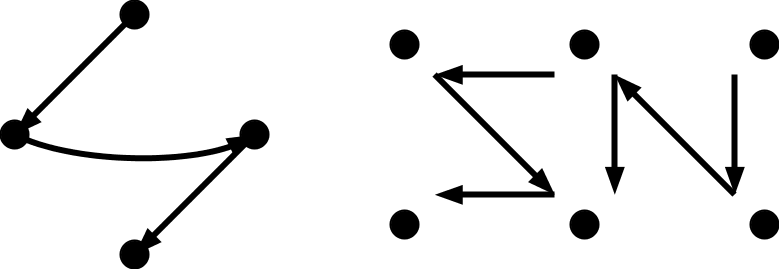}
\caption{Scheduling for the standard (left) and rotated (right) subspace surface codes.  On the rotated lattice, the left plaquette represents an $X$-type stabilizer, while the right represents a $Z$-type stabilizer.  These are chosen so that hook errors occur parallel to the boundary of the same type: in this case, an $X$-type north-south boundary and a $Z$-type east-west boundary.}
\label{subspace_schedulings}
\end{figure}

For the subsystem surface code, we use the trivial scheduling; see Figure \ref{subsystem_scheduling}.  It is straightforward to check that the syndrome extraction satisfies the goodness property.  For correctness, as the $X$- and $Z$-type syndromes are extracted serially, there can be no randomization due to anticommuting operators.

 \begin{figure}[htb!]
\includegraphics[width=.55\linewidth]{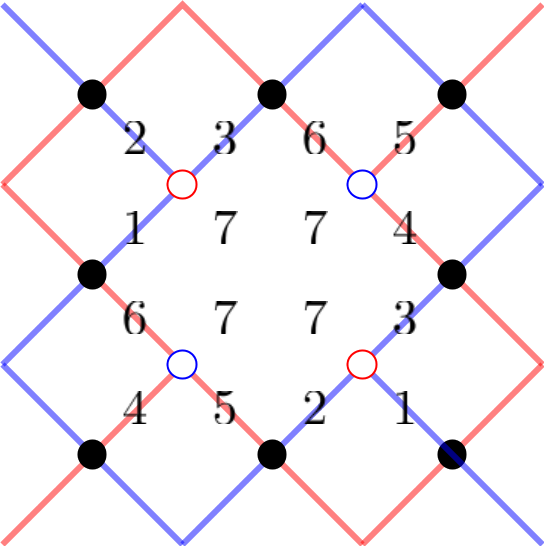}
\caption{Trivial scheduling for the subsystem code.  Dark circles are code qubits, red silhouetted circles are $X$-ancilla, and blue silhouetted circles are $Z$-ancilla.  The numbers refer to CNOT locations in a particular time-step, except for those in the center that denote measurement. Full syndrome extraction proceeds in seven time-steps (not including ancilla preparation). }
\label{subsystem_scheduling}
\end{figure}

However, we demonstrate that there also exists a good and correct syndrome extraction scheme that is fully-parallelized, i.e., there are no idling data qubits in any time-step.  For this, we will need a lemma from \cite{bravyi2013subsystem} giving criteria on when a syndrome extraction scheduling is correct.

\begin{Lemma}[Lemma 1, \cite{bravyi2013subsystem}]\label{bravyi_lemma}
A fully-parallelized schedule of CNOTs is correct if for any pair of $X$-type and $Z$-type triangular gauge operators, at least one of the following is true:
\begin{enumerate}
    \item[$(i)$] the two operators are disjoint;
    \item[$(ii)$] the two operators are measured two rounds apart;
    \item[$(iii)$] the last gate of the operator measured first in the pair commutes with the first gate of the operator measured last in the pair.
\end{enumerate}
\end{Lemma}

For proof, see \cite{bravyi2013subsystem}.  One can show that any scheduling which is translation invariant with respect to a single pair of overlapping hexagonal stabilizers, as shown in Figure \ref{subsystem_scheduling}, cannot simultaneously satisfy Lemma \ref{bravyi_lemma} and satisfy goodness.  

\begin{figure}[htb!]
\includegraphics[width=0.74\linewidth]{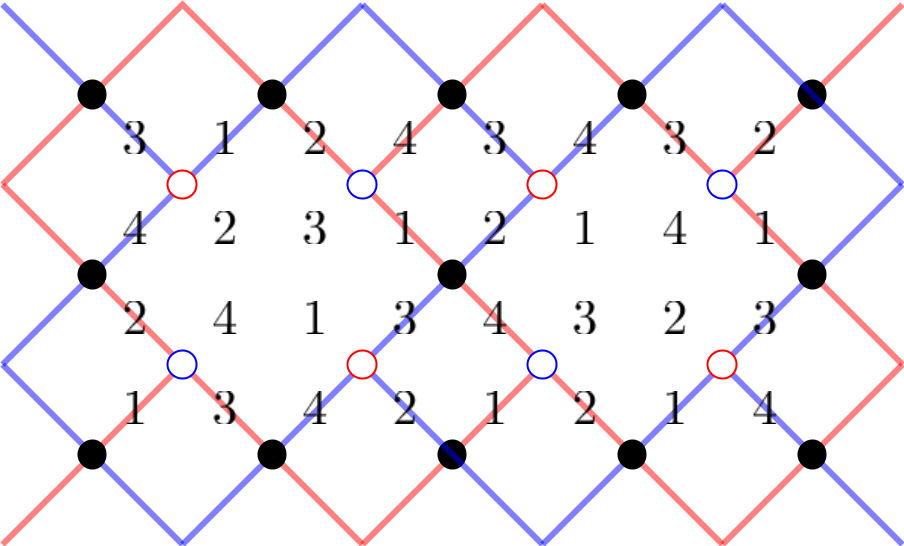}
\caption{Fully-parallelized good and correct scheduling.  Rolling syndrome extraction is performed over four different time-steps, with the middle numbers indicating when measurements are occurring.  A full scheduling can be realized as translations of this fundamental cell.  In the presence of leakage, one must insert swapping of the data with auxiliary qubits into the scheduling.}
\label{subsystem_scheduling_hard}
\end{figure}

However, there exist such schedulings that are translation invariant on adjacent pairs of overlapping hexagonal stabilizers, and such a scheduling suffices for odd distance codes with boundary.  See Figure \ref{subsystem_scheduling_hard} for an example, whose goodness can be checked and which satisfies the hypotheses of Lemma \ref{bravyi_lemma}.  This example demonstrates that there exist fully-parallelized syndrome extractions that are both good and correct.  Note that, although there are no idling data qubits, there will be (few) idling time-steps for the ancilla qubits as data qubits are swapped with auxiliary qubits.

Finally, for a swap-$LR$ scheme defined on subsystem surface codes, long-range correlated errors may occur when triangular gauge operators of the same type interact with a data qubit in any combination of time-steps two and three.  This can be avoided  for all but one data qubit (up to translation-symmetry), with every data qubit swapping with some ancilla in time-step three.  If that one data qubit is chosen properly, these long-range correlated errors do not damage the effective distance of the code. See Figure \ref{swap_scheduling} for the scheduling and the left-hand side of Figure \ref{bad_data_leakage} for the benign correlated error. 

\begin{figure}[htb!]
\includegraphics[width=0.55\linewidth]{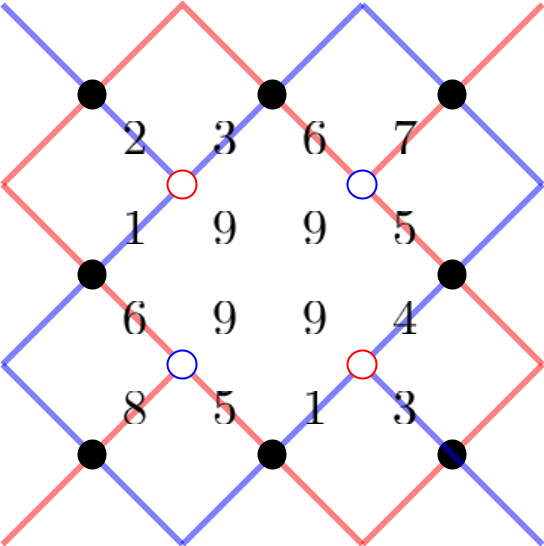}
\caption{Subsystem code scheduling for swap-$LR$.  The last gate in any triangular gauge operator swaps data and ancilla, except for the northeast ancilla qubit.  The northwest/southeast data qubit may cause a long-range $X$-type correlated error as it interacts with two ancilla qubits in time-step two.  However, unlike the damaging data leakage on the right-hand side of Figure \ref{bad_data_leakage}, these correlated errors may also be formed by a single $Z$-type ancilla leakage, up to gauge transformation.  By our classification of ancilla leakage errors, they cannot reduce the effective distance of the code.}
\label{swap_scheduling}
\end{figure}

\section{Auxiliary Qubit Leakage Reduction} \label{auxiliary}

There are many advantages to using minimal leakage reduction; for example, preserving the connectivity of the qubit lattice.  However, one could also consider strategies which involve many auxiliary qubits and LRUs.  We detail two such additional leakage elimination methods here.

\begin{enumerate}
    \item[$(i)$] \emph{Gate leakage reduction} (gate-$LR$) swaps the qubits involved in each gate with auxiliary qubits immediately after the gate is applied.  A single leakage may only persist for a single time-step, effectively converting the leakage model to a depolarizing model at the cost of significant overhead.  
    \item[$(ii)$] \emph{Intermediate leakage reduction} (int-$LR$) applies syndrome-$LR$, as well as swapping with an auxiliary qubit after the second CNOT in each syndrome extraction. For subspace surface codes, it swaps both the data and the ancilla with auxiliary qubits, while for subsystem surface codes, it swaps only the ancilla.
\end{enumerate}

Gate-$LR$ was studied in \cite{suchara2015leakage} as a means of retaining the effective code distance.  However, its tremendous overhead caused much lower thresholds, around $p_{\text{thr}} = 1.6 \times 10^{-3}$ when translated to our noise model. The higher thresholds of low-overhead subsystem code leakage reduction indicate better performance given that both preserve the effective distance of the code, but also depend on the assigned cost of the required auxiliary qubits.  By direct comparison, Bacon-Shor codes also perform better at low distances. 

\begin{figure}[t!]
\includegraphics[width=0.8\linewidth]{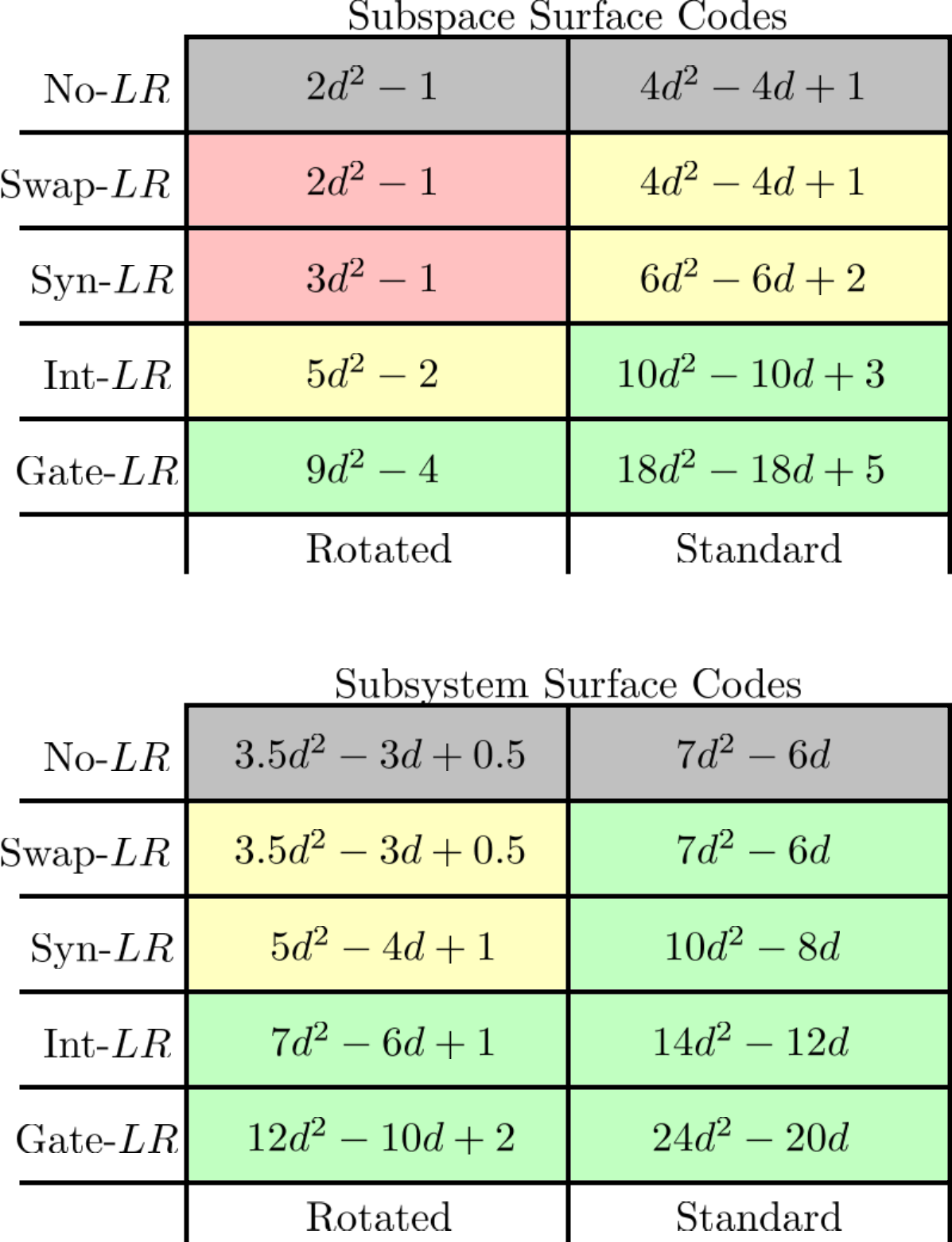}
\caption{A summary of the overall qubit requirements for different leakage reduction strategies, where $d$ is the distance of the code.  Red represents an effective distance reduction in both $MS$- and $DP$-leakage, yellow in only $DP$-leakage, green in neither, while grey indicates no threshold.  Motivated by qubit timings in ion-traps \cite{trout2018simulating}, each cell indicates the total qubit count (including data and ancilla) required in a single round of syndrome extraction.  The effective distance for swap-$LR$ is estimated numerically rather than analytically, with auxiliary qubits included for syndrome-$LR$.}
\label{Overheads}
\end{figure}

Intermediate leakage reduction is a more intermittent leakage elimination strategy, which can preserve the effective code distance with lower overhead but allows for more correlated errors compared to gate-$LR$.  We summarize these relationships in Figure \ref{Overheads}, which measures the circuit-volume overhead of each strategy.  Note that we do not include the extra $\theta(d)$ qubits required on the boundary of the qubit lattice in order to preserve the connectivity using swap-$LR$.

\section{Non-Topological Counterexample} \label{topcounterexample}

Here we illustrate a non-topological counterexample to Proposition \ref{main}.  It reflects a counter-intuitive point in the proof: the upper bound on $d_{\text{eff}}$ grows tighter as the diameter of the stabilizer decreases.  One would expect that larger stabilizers would give rise to larger correlated errors, and thus would give tighter upper bounds on $d_{\text{eff}}$. 

However, large overlapping stabilizer generators may prevent bad correlated errors from combining.  In fact, by relaxing the restriction of a topological generating set, we can manage $d_{\text{eff}} = d - 1$ in the surface code when decoding is performed carefully \cite{li20182}.  If we label each plaquette $i$ rows from the top and $j$ columns from the left as $P_{ij}$, and similarly for vertices $V_{ij}$, then this generating set may be realized as $$\left\langle\left\{\{\prod\limits_{k=0}^j X_{P_{ik}} \}_{j=0}^{d-2} \right\}_{i=0}^{d-1}, \left\{\{\prod\limits_{k=0}^i Z_{V_{kj}} \}_{i=0}^{d-2} \right\}_{j=0}^{d-1} \right\rangle.$$ Still, this is not a practical strategy as it requires measuring long-range stabilizers that scale with the lattice size, see Figure \ref{counterexample}. While this will unreasonably degrade code performance, it demonstrates that the direction of the upper bound and the hypotheses of Proposition \ref{main} are necessary.

\begin{figure}[htb!]
\includegraphics[width=0.85\linewidth]{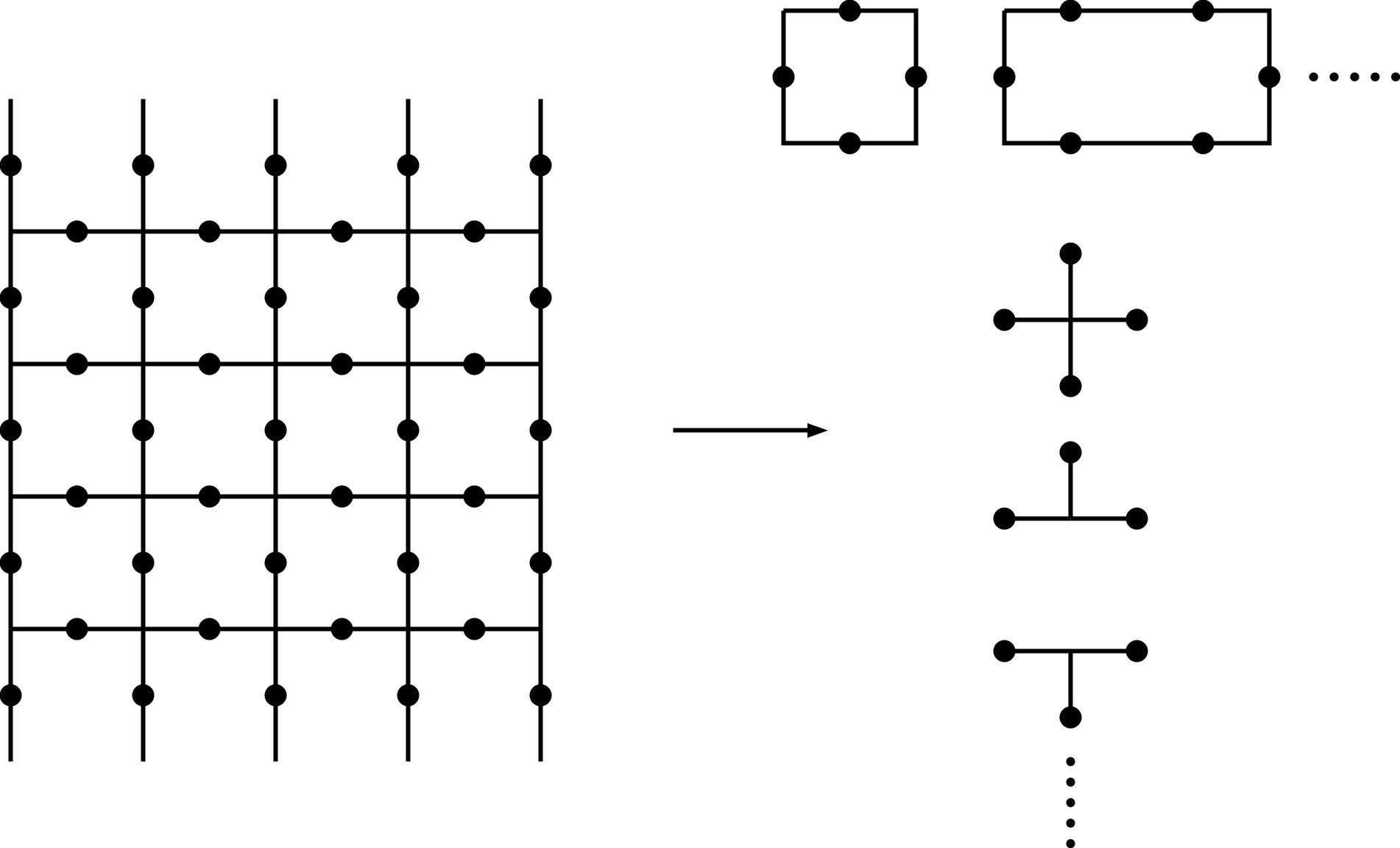}
\caption{A non-topological generating set that only reduces the effective code distance by one.  Each $X$-stabilizer is measured as an increasing product of plaquettes along each row, and each $Z$-stabilizer is measured as an increasing product of dual plaquettes along each column.  Note, for example, that a single leakage may produce an $X$-type error that spans the dual lattice from north to south, but this cannot produce a logical error.}
\label{counterexample}
\end{figure}

\section{Threshold Estimates} \label{threshold_estimates}
In this section, we directly compare subsystem surface codes to subspace surface codes, while accounting for the $1.75\times$ qubit overhead required for subsystem codes of the same distance and lattice geometry. Although we have focused on the low $p$ limit at low distances to remain simulable, we are more interested in the error-suppression at fixed $p$ for increasing $d$.  Up to a polynomial pre-factor, we can express $p_L(p,d) \sim \exp(-\alpha(p)d)$.  Determining the $p$ for which $\alpha_{\text{subsystem}}(p) > \alpha_{\text{subspace}}(p)$, optimized with respect to both the qubit overhead and geometry, will determine the error regime in which subsystem surface codes absolutely outperform subspace surface codes.

Unfortunately, in the low-error regime we consider, it is difficult to obtain accurate error estimates at even intermediate distances, and low distance estimates suffer from finite-size effects.  However, we can make heuristic estimates via the high-distance fit \cite{fowler2012towards}, \begin{equation}\label{myeqn}p_L(p,d) \sim \left(\frac{p}{p_{\text{thr}}}\right)^{d_e},\end{equation} where $p_{\text{thr}}$ is the threshold and $d_e$ is the estimated scaling, based on the near-fitting of Figure \ref{main:sim} to the minimum number of faults required to produce a logical error using syndrome-$LR$.  We can estimate the thresholds accurately as they occupy a higher-error regime, see Figure \ref{Thresholds}.

\begin{figure}[t!] 
\centering
\begin{subfigure}{0.25\textwidth}
\includegraphics[width=\linewidth]{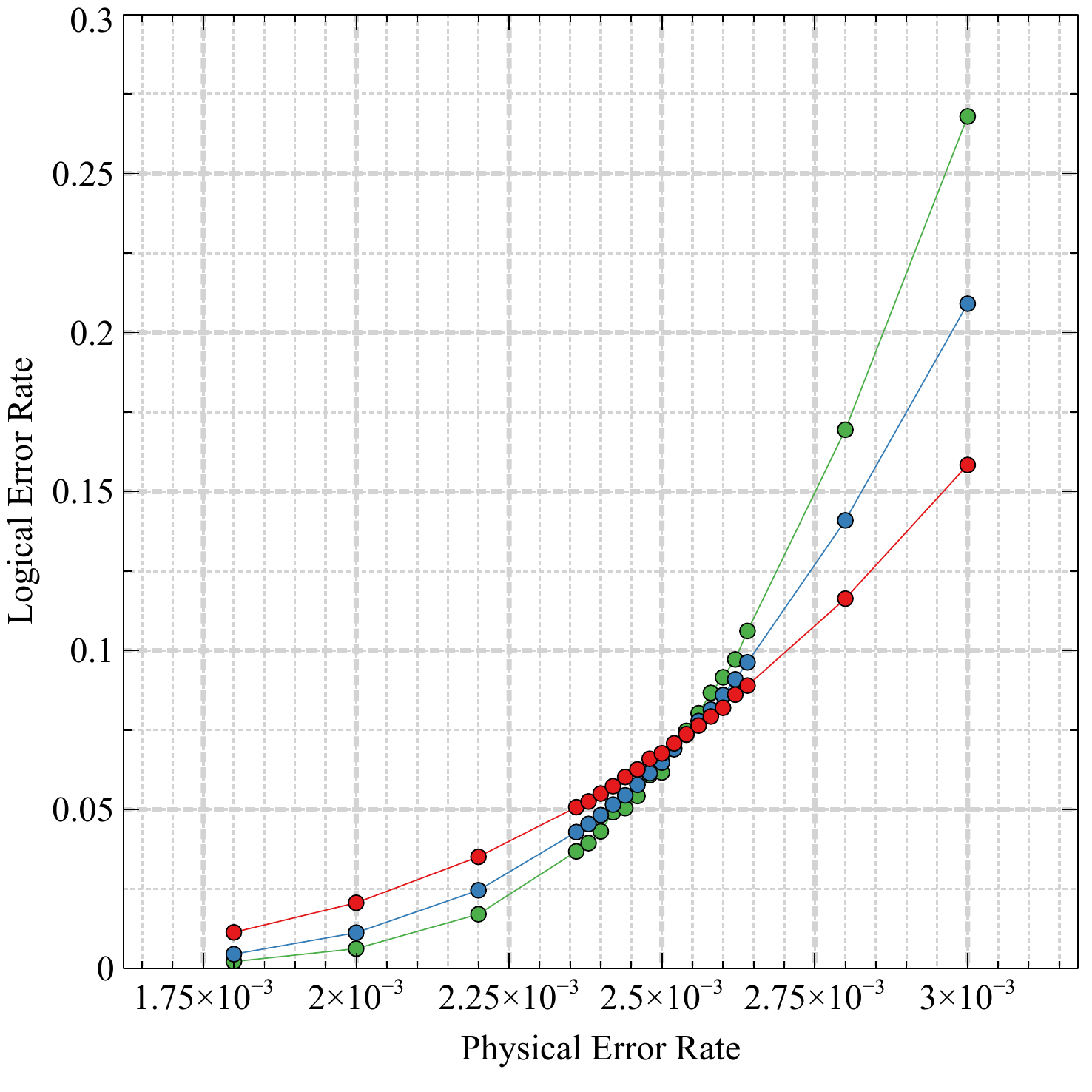}
\caption{$p_g^{DP} \approx 2.55 \times 10^{-3}$} \label{fig:asdg23}
\end{subfigure}\hspace*{\fill}
\begin{subfigure}{0.25\textwidth}
\includegraphics[width=\linewidth]{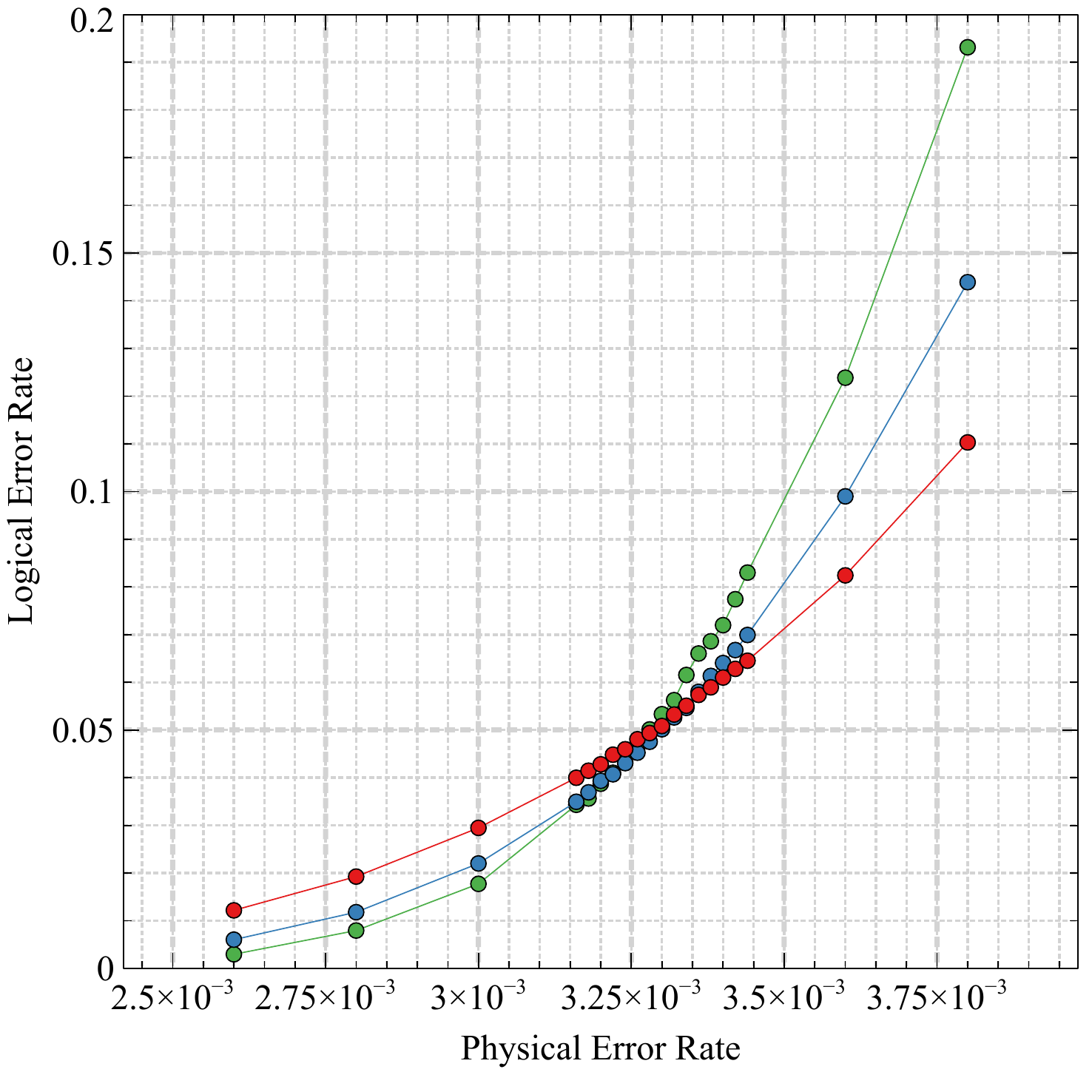}
\caption{$p_g^{MS} \approx 3.30 \times 10^{-3}$} \label{fig:b3asf4}
\end{subfigure}

\medskip
\begin{subfigure}{0.25\textwidth}
\includegraphics[width=\linewidth]{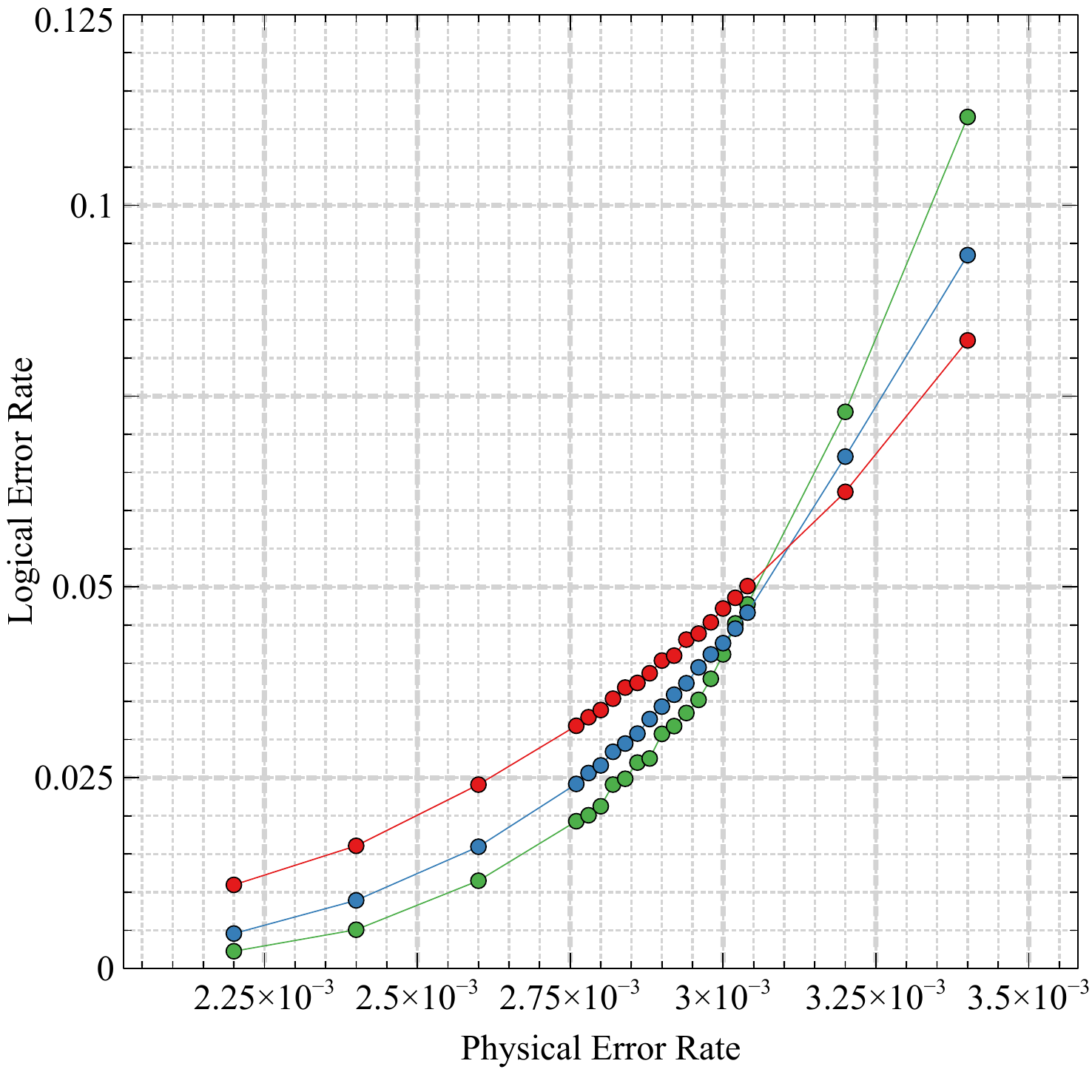}
\caption{$p_s^{DP} \approx 3.00 \times 10^{-3}$} \label{fig:casf45}
\end{subfigure}\hspace*{\fill}
\begin{subfigure}{0.25\textwidth}
\includegraphics[width=\linewidth]{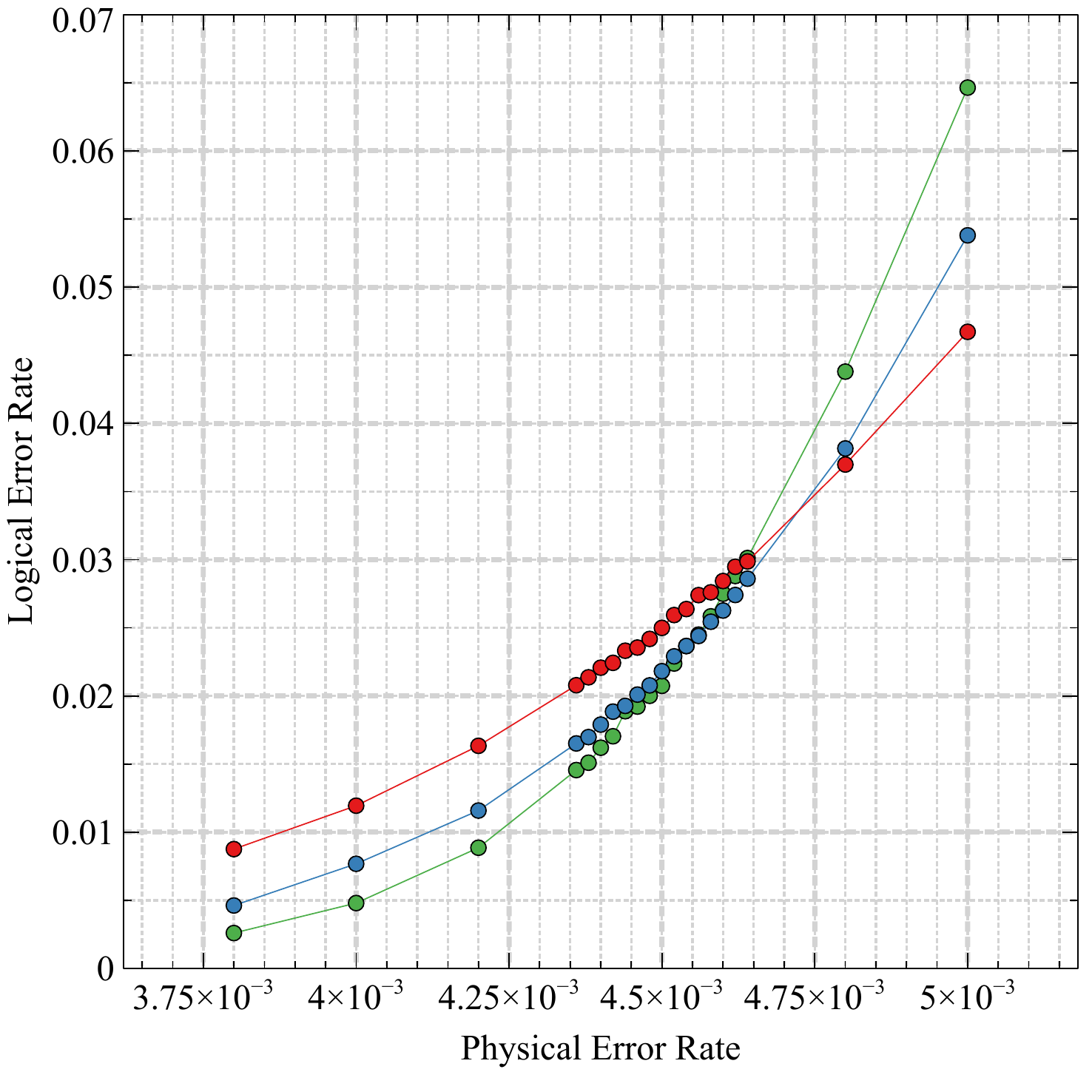}
\caption{$p_s^{MS} \approx 4.55 \times 10^{-3}$} \label{fig:d5asfd6}
\end{subfigure}

\caption{Thresholds simulated from distance $9,13,$ and $17$ lattices with periodic boundaries to minimize finite-size effects. The left column and right column represent $DP$- and $MS$-leakage, respectively, while the top row and bottom row represent subsystem and subspace codes, respectively.  Each data point was obtained from $5 \times 10^4 - 5 \times 10^5$ trials, with distance-$9$ lattices still exhibiting some finite-size effects.} \label{Thresholds}
\end{figure}

Given thresholds $p_g,p_s$ for a subsystem and subspace code, respectively, and normalizing qubit-overhead factors $\gamma_g,\gamma_s$ taken from Table \ref{Overhead}, we can upper-bound the error rate $p_*$ below which subsystem surface codes outperform subspace surface codes at sufficiently high distance from Equation \ref{myeqn}, $$p_* \sim \left(\frac{p_g^{\gamma_g}} {p_s^{\gamma_s}}\right)^ {\frac{1}{\gamma_g - \gamma_s}}.$$
This estimate yields $p_*^{DP} \approx 2.5 \times 10^{-4}$ in the presence of $DP$-leakage and $p_*^{MS} \approx 0.32 \times 10^{-4}$ in the presence of $MS$-leakage at sufficiently high distance.  Given Table \ref{Overhead}, these $p_*$ are obtained by comparing the optimal subsystem and subspace code choices for each leakage type.  These are again the standard subsystem codes and rotated subspace codes in the presence of $DP$-leakage, and the rotated subsystem codes and standard subspace codes in the presence of $MS$-leakage.

 Generally, fitting formulas like Equation \ref{myeqn} tend to underestimate the logical error rate at finite $p$ \cite{bravyi2013simulation}.  This, coupled with the observation that Pauli errors contribute more heavily at near-threshold errors rates, makes it likely that these fits already favor the subsystem surface codes, and hence only provide a heuristic upper-bound.

 It is worth noting that our simulations actually demonstrate better low-rate scaling with rotated subsystem codes than with standard subspace codes in the presence of $MS$-leakage, somewhat counteracting the overestimation of $p_*^{MS}$.  Further numerics are required to generate a more accurate estimate based on empirical scalings at higher distances.
 
 \section{Correlated Leakage} \label{correlated}
 
 For certain architectures like quantum dots, correlated leakage events may occur simultaneously on the full support of a two-qubit gate.  In this case, leakage robustness follows directly from considering correlated error patterns using swap-$LR$, where instead of leakage occurring on ancilla and data qubits in consecutive rounds, they occur simultaneously in the same round.  Then the same set of space-correlated errors will occur as long as one instead uses syndrome-$LR$.  Note that if one uses swap-$LR$ directly, then in the case of $MS$-leakage, we may depolarize two qubits in the support of the measurement: the data qubit that is leaked and the data qubit that is swapped with the leaked ancilla.
 
 We compare Bacon-Shor codes and surface codes with two-qubit cat state extraction in the presence of correlated leakage.  Here, preparation of the cat state itself may cause a damaging error.  The two codes display similar relative behavior when using bare-ancilla extraction with independent leakage, see Figure \ref{corr_data}.  By the preceding argument, subsystem surface codes should also achieve the correct logical error suppression.
 
 \begin{figure}[htb!]
\includegraphics[width=.8\linewidth]{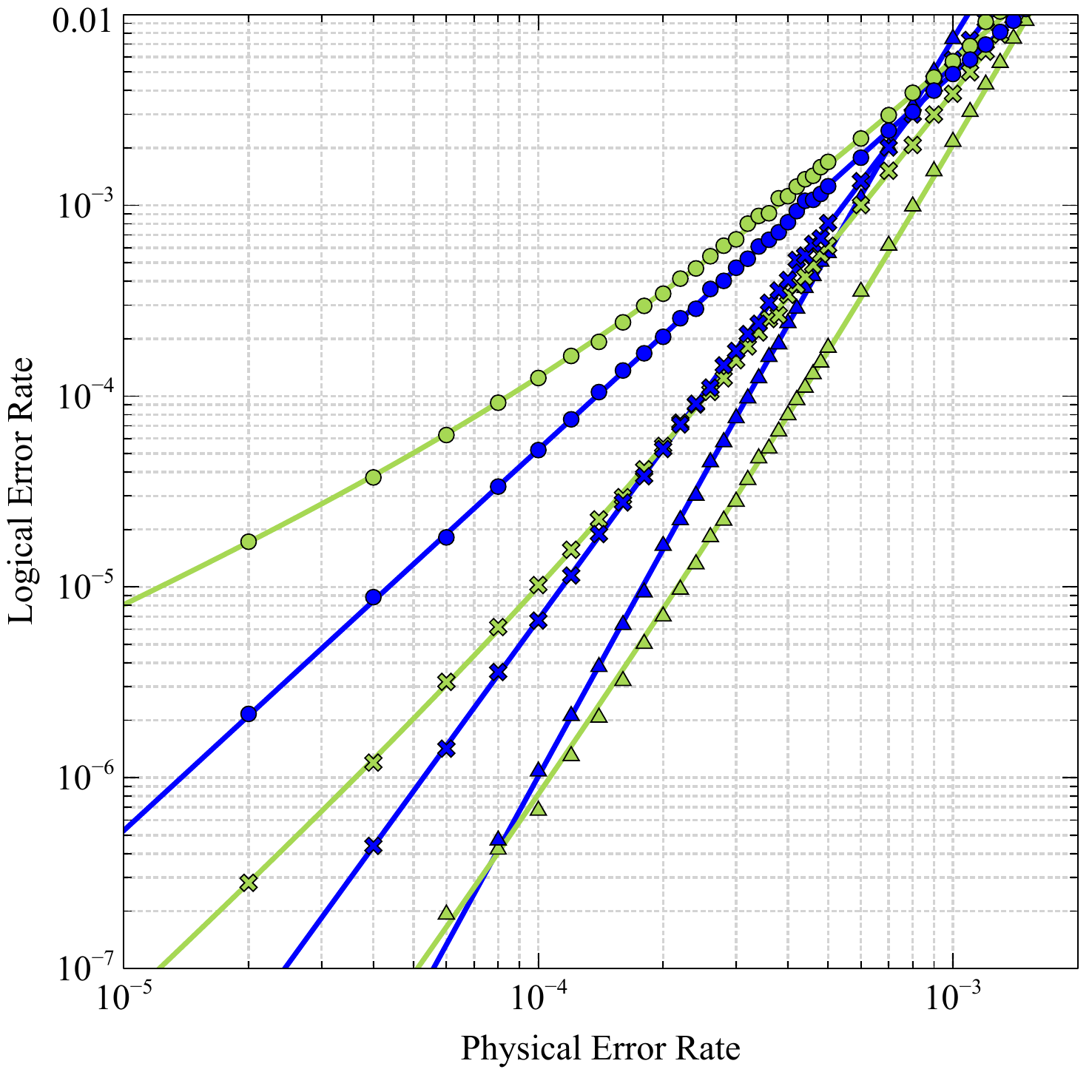}
\caption{Distance $3,5,$ and $7$ rotated subspace surface codes (green) compared with Bacon-Shor codes (blue) in the presence of correlated $MS$-leakage.  Bacon-Shor codes begin to demonstrate an advantage at errors rates $\lessapprox 10^{-3}$ at distance-$3$.  This advantage drops off quickly at higher distances.}
\label{corr_data}
\end{figure}

\end{document}